\documentclass[english]{lipics-v2021}

\usepackage{caption}
\usepackage{makeidx}
\usepackage{graphicx}
\usepackage{graphics}
\usepackage{epsfig}
\usepackage{epstopdf}
\usepackage{amsmath,amssymb}
\usepackage{latexsym}
\usepackage{algorithm}
\usepackage{algorithmicx}
\usepackage{mathtools}
\usepackage{thmtools}
\usepackage{thm-restate}
\usepackage{placeins}
\usepackage{wrapfig}
\usepackage{multicol}

\newcommand{\deleted}[1]{}
\providecommand{\norm}[1]{\lVert#1\rVert}

\usepackage{microtype}

\nolinenumbers

\title{In-Range Farthest Point Queries and Related Problem in High Dimensions}
\titlerunning{In-Range Farthest Point Queries in High Dimensions}

\author{Ziyun Huang}{Department of Computer Science and Software Engineering,\\
Penn State Erie, The Behrend College, Erie, United States}{zxh201@psu.edu}{}{}

\author{Jinhui Xu}{Department of Computer Science and Engineering,\\
State University of New York at Buffalo, Buffalo, United States}{jinhui@buffalo.edu}{}{The research of this author was supported in part by NSF through grant IIS-1910492 and by KAUST through grant CRG10 4663.2.}

\authorrunning{Z. Huang and J. Xu}

\Copyright{Ziyun Huang and Jinhui Xu}


 
\EventEditors{}
\EventNoEds{0}
\EventLongTitle{}
\EventShortTitle{}
\EventAcronym{}
\EventYear{2022}
\EventDate{ }
\EventLocation{ }
\EventLogo{}
\SeriesVolume{ }
\ArticleNo{ }

\begin{document}

\maketitle

\begin{abstract}
Range-aggregate query is an important type of queries with numerous applications. It aims to obtain some structural information (defined by an \emph{aggregate function} $F(\cdot)$) 
of the points (from a point set $P$) inside a given query range $B$. 
In this paper, we study the range-aggregate query problem in high dimensional space for two  aggregate functions:
(1) $F(P \cap B)$ is the farthest point in $P \cap B$ to a query point $q$ in $\mathbb{R}^d$ and
(2) $F(P \cap B)$ is the minimum enclosing ball (MEB) of $P \cap B$.
For 
problem (1), called \emph{In-Range Farthest Point (IFP) Query}, 
we develop a bi-criteria approximation scheme: For any  $\epsilon>0$ that specifies the approximation ratio of the farthest distance
and any $\gamma>0$ that measures the ``fuzziness'' of the query range, we show that it is possible to pre-process  
$P$ into a data structure of size $\tilde{O}_{\epsilon,\gamma}(dn^{1+\rho})$ in $\tilde{O}_{\epsilon,\gamma}(dn^{1+\rho})$ time such that given any $\mathbb{R}^d$ query ball $B$ and query point $q$, it outputs in $\tilde{O}_{\epsilon,\gamma}(dn^{\rho})$ time a point $p$ that is a $(1-\epsilon)$-approximation of the farthest point to $q$ among all points lying in a  $(1+\gamma)$-expansion $B(1+\gamma)$ of $B$,    
where 
$0<\rho<1$  is a constant depending on 
$\epsilon$ and $\gamma$
and the hidden constants in big-O notations depend only on $\epsilon$, $\gamma$ and $\text{Polylog}(nd)$. 
For problem (2),
we show that the IFP result can be applied to
develop query scheme with similar time and space complexities to achieve a $(1+\epsilon)$-approximation for MEB.
To the best of our knowledge, these are 
the first theoretical results on such high dimensional range-aggregate query problems. 
Our results are based on several new techniques, such as {\em multi-scale construction} and {\em ball difference range query},  
which are interesting in their own rights and could be potentially used to solve other range-aggregate problems in high dimensional space. 
\end{abstract}


\section{Introduction}

Range search is a fundamental problem in computational geometry and finds applications in many fields like database systems and data mining \cite{agarwal1999geometric,tao2007range}. It has the following  basic form:  Given a set of $n$ points $P$ in $\mathbb{R}^d$, pre-process $P$ into a data structure  so that for any query range $B$ from a certain range family (\emph{e.g.,} spheres, rectangles, and halfspaces), it reports or counts the number of the points in $P \cap B$ 
efficiently.
Range search allows us to obtain some basic information of the points that lie
in a specific local region of the space.

In many applications, it is often expected to know more information than simply the number of points in the range. This leads to the study of range-aggregate query 
\cite{abr07-rcluster,agarwal2013efficient,arya2015approximate,brass2013range,gupta2014data,kha14-qmeb-improved,li2020polyfit,nek10-qmeb,rah10-range,rah11-range,yun2014fastraq}, 
which is a relatively new type of range search.  
The goal of range-aggregate query is to 
obtain more complicated structural information (such as the diameter, the minimum enclosing ball, and the minimum spanning tree) of the points
in the query range. 
Range-aggregate query can be generally defined as follows: Given a point set $P$, pre-process $P$ into a data structure such that for any range $B$ in a specific family, it
outputs $F(P \cap B)$,
where $F(\cdot)$ is a given \emph{aggregate function}
that computes a certain type of information or structure of $P \cap B$ like ``diameter'',``minimum enclosing ball'', and ``minimum spanning tree''.
Range-aggregate queries have some interesting applications in data analytics and big data \cite{ho1997range,vitter1999approximate,wu2013scorpion,yun2014fastraq}, where it is often required to retrieve aggregate information of the records in a dataset with keys that lie in any given (possibly high dimensional) range. 

In this paper, we study the range-aggregate query problem in high dimensions
for spherical ranges. Particularly, 
we consider two aggregate functions for any $\mathbb{R}^d$ query ball $B$:
(1) $F(P \cap B)$ is the farthest point in $P \cap B$ to a query point $q$
in $\mathbb{R}^d$ and (2) $F(P \cap B)$ is the minimum enclosing ball (MEB) of $P \cap B$.
We will focus in this paper on 
problem (1), called the
\emph{In-Range Farthest Point (IFP) Query}, and 
show that an efficient solution to IFP query also yields efficient solutions to
the MEB 
problems.
We start with some definitions.   

\begin{definition} 
({\bf Approximate IFP (AIFP)})
Let $P$ be a set of $n$ points in $\mathbb{R}^d$,
$q$ be a point and $B$ be a $d$-dimensional (closed) ball. 
A point $p \in P$ is a bi-criteria $(\epsilon, \gamma)$-approximate in-range farthest point (or AIFP) of $q \in P$ in $B$,
if there exists a point set $P'$ such that the following holds, where  $\epsilon$ and $\gamma$ are small positive constants,
and $B(1+\gamma)$ is the ball concentric with $B$ and with radius $(1+\gamma)r$:  
(1) $P \cap B \subseteq P' \subseteq P \cap B(1+\gamma)$; 
(2) $p \in P'$;
and 
(3) for any $p' \in P'$,  $(1-\epsilon) \norm{p' - q} \leq \norm{p - q}$.
\end{definition}

Defining AIFP in this way enables us to consider all points in $B$ and exclude all points outside of $B(1+\gamma)$. Points in the fuzzy region $B(1+\gamma) \setminus B$ may or may not be included in the farthest point query. 
Note that allowing fuzzy region is a commonly used strategy to deal with the challenges in many high dimensional similarity search and range query problems.
For example, consider the classic near neighbor search problem, which is
equivalent to spherical emptiness range search:
Given a query sphere $B$ in $\mathbb{R}^d$, report a data point $p$ that lies in $B$ if such a data point exists.
In high dimensional space, obtaining an exact solution to such a query is very difficult.
A commonly used technique for this problem is the Locality Sensitive Hashing (LSH) scheme \cite{datar2004locality}.
Given a query ball $B$, LSH could report a data point in $B(1+\epsilon)$ for some given factor $\epsilon > 0$.
In other words, a fuzzy region $B(1+\epsilon) \setminus B$ is allowed.
Similarly, we can define 
approximate MEB for points
in a given range with a fuzzy region.

\begin{definition}
(\textbf{Minimum Enclosing Ball (MEB)})
Let $P$ be a set of $n$ points in $\mathbb{R}^d$. 
A $d$-dimensional (closed) ball $B$ is an enclosing ball of $P$ if $P \subset B$ and $B$ is the minimum enclosing ball (MEB) of $P$ if its radius $r$ is the smallest among all enclosing balls.
A ball $B'$ is a $(1+\epsilon)$-approximate MEB of $P$ for some constant $\epsilon > 0$ if it is an enclosing ball of $P$ and its radius is no larger than $(1+\epsilon)Rad(P)$, where $Rad(P)$ is the radius of the MEB of $P$.  
\end{definition}

\begin{definition} 
({\bf Approximate MEB (AMEB)})
Let $P$ be a set of $n$ points and $B$ be any ball with radius $r$ in $\mathbb{R}^d$. A ball $B'$ with radius $r'$ is a bi-criteria 
$(\epsilon, \gamma)$-approximate MEB (or AMEB) of $P$ in range $B$,
if there exists a point set $P'$ such that the following holds, where $\gamma$ and $\epsilon$ are small positive constants:  
    (1) $P \cap B \subseteq P' \subseteq P \cap B(1+\gamma)$; and 
    (2) $B'$ is a $(1+\epsilon)$-approximate MEB of $P'$.
\end{definition}

In this paper, we will focus on 
building a data structure for $P$ 
so that given any query ball $B$ and a point $q \in \mathbb{R}^d$,
an AIFP of $q$ in $P\cap B$ can be computed efficiently ({\em i.e.,} in sub-linear time in terms of $n$). 
Below are the main theorems of this paper.
Let $\epsilon > 0,\gamma > 0$, $0 < \delta < 1$ be any real numbers.

\begin{theorem}
\label{thm-main-aipf}
For any set $P$ of $n$ points in $\mathbb{R}^d$,
it is possible to build a data structure of 
size $O_{\epsilon,\gamma}(d n^{1+\rho} \log \delta^{-1}\text{Polylog}(nd))$ in  
$O_{\epsilon,\gamma}(dn^{1+\rho} \log \delta^{-1} \text{Polylog}(nd))$ pre-processing time, 
where $0<\rho<1$ is a small constant depending on
$\epsilon$ and $\gamma$.
With this data structure, it is then possible to find a $(\epsilon, \gamma)$-AIFP of any given query point $q$ and
query ball $B$
in 
$O_{\epsilon,\gamma}(dn^{\rho} \log \delta^{-1} \text{Polylog}(nd))$ time with probability at least $1 - \delta$.
\end{theorem}

{\em Note: In the above result, the relationship between $\rho$ and $\epsilon, \gamma$ has a rather complicated dependence on several constants of $p$-stable distribution, which is inherited from the  underlying technique of Locality Sensitive Hashing (LSH) scheme \cite{datar2004locality}. 
This indicates that for any $\epsilon,\gamma$, we have $0<\rho<1$ and $\rho$ approaches $1$ as $\epsilon,\gamma$ approach $0$.
}

We will also show how to use the AIFP data structure to answer MEB queries efficiently. 

\begin{theorem}
\label{thm-main}
For any set $P$ of $n$ points in $\mathbb{R}^d$,
it is possible to build a data structure of size 
$O_{\epsilon,\gamma}(d n^{1+\rho} \log \delta^{-1} \text{Polylog}(nd))$ in  
$O_{\epsilon,\gamma}(dn^{1+\rho} \log \delta^{-1} \text{Polylog}(nd))$ pre-processing time, 
where $0<\rho<1$ is a small constant depending on
$\epsilon$ and $\gamma$.
With this data structure, it is then possible to find a  $(\epsilon, \gamma)$-AMEB for any query ball $B$ in 
$O_{\epsilon,\gamma}(dn^{\rho} \log \delta^{-1} \text{Polylog}(nd))$ time with 
probability at least $1 - \delta$.
\end{theorem}

To our best knowledge, 
these are the first results on such range-aggregate problems
in high dimensions.
Each data structure has only a near linear dependence on $d$, a sub-quadratic dependence on $n$ in space complexity, and a sub-linear dependence on $n$ in query time. 

\vspace{0.1in}

\noindent\textbf{Our Method:} The main result on AIFP is based on several novel techniques, such as {\em multi-scale construction} and {\em ball difference range query}.
Briefly speaking,
multi-scale construction is a general technique that allow us to
break the task of building an AIFP query data structure into a number of ``constrained'' data structures. Each such data structure is capable of correctly answering an AIFP query given that some assumption about the query holds (for example, the distance from $q$ to its IFP is within a
certain range). 
Multi-scale construction uses a number of
``constrained'' data structures of small size to cover all possible cases of a query, which leads to a data structure that can handle any arbitrary queries. 
Multi-scale construction is independent of the aggregate function, and thus has the potential be used as a general method for other types of range-aggregate query problems in high dimensional space. 
Another important technique is a data structure for the ball difference range query problem, which returns a point, if there is one, in the difference of two given query balls. 
The ball difference data structure is the building block for the constrained AIFP data structures,
and is interesting in its own right as a new high dimensional range search problem.

\noindent\textbf{Related Work:} 
There are many results for the ordinary farthest point query problem in
high dimensional space
 \cite{curtin2016fast,huang2017-far,indyk2003better,pagh2015approximate}. 
However, to the best of our knowledge,
none of them
is sufficient to solve the IFP problem, 
and our result is the first one to consider the farthest point problem under the query setting. 
Our technique for the IFP problem also yields solutions to other range-aggregate queries problems,
including the MEB query problem.

A number of results exist for
various types of the range-aggregate query problem in fixed dimensional space.
In \cite{arya2015approximate}, Arya, Mount, and Park proposed an elegant scheme for querying minimum spanning tree inside a query range. 
They showed that there exists a bi-criteria $(\epsilon_q,\epsilon_w)$-approximation with a query time of $O(\log n + (1/\epsilon_q\epsilon_w)^d)$.
In \cite{nek10-qmeb}, Nekrich and Smid introduced a data structure to compute an $\epsilon$-coreset for the case of  orthogonal query ranges and aggregate functions  satisfying some special properties. 
Xue \cite{xue2019colored} considered the colored closest-pair problem in a (rectangular) range and obtained a couple of data structures with near linear size and polylogarithmic query time. 
Recently,  Xue \emph{et. al.} \cite{xue2020approximate} further studied more general versions of the closest-pair problem
and achieved similar results.
For the MEB problem under the range-aggregate settings,  Brass {\em et al.}  are the first to investigate the problem in 2D space,  along with other types of aggregate functions (like width and the size of convex hull) \cite{brass2013range}. 
They showed that it is possible to build a data structure with 
$O(n \cdot polylog(n))$ pre-processing space/time and $O(polylog(n))$ query time.

All the aforementioned methods were designed for fixed dimensional space, and thus are not applicable to high dimensions.
Actually, range aggregation has rarely been considered in high dimensions, except for a few results that may be viewed as loosely relevant. 
For example,  Abbar {\em et al.} \cite{abbar2013chishi} studied the problem of finding the maximum diverse set for points inside a ball with  fixed radius around a query point. Their ideas are  seemingly useful to our problem. However, since  their ball always has the same fixed radius, 
their techniques are not directly applicable. In fact,   
a main technical challenge of our problem is how to deal with the arbitrary radius and location of the query range, which is overcome by our multi-scale construction framework. 
Another related work by Aum{\"u}ller \emph{et. al.} \cite{aumuller2021sampling}
has focused on random sampling in a given range.
The technique is also not directly applicable to IFP.

\subsection{Overviews of the Main Ideas}

Below we describe the main ideas of our approaches. For simplicity,
in the following we ignore the fuzziness of the query range. 
We approach the AIFP query problem by first looking at an easier version:
given ball $B$ and point $q$, find an approximate farthest point in $P \cap B$ to $q$,
with the (strong) assumption that the radius of $B$ is a fixed constant $r_{B} > 0$,
and that the distance between $q$ and its IFP in $P \cap B$ is within a range  of $( d_{min}, d_{max} ]$,
where $d_{max} > d_{min} > 0$ are fixed constants.
We call such a problem a constrained AIFP problem.
We use a tuple $(r_{B}, d_{min}, d_{max})$ to denote such a constraint.

To solve the constrained AIFP problem,
we develop a data structure for the 
 \emph{ball difference (BD) range query} problem, which is defined as follows: given two balls $B_{in}$ and $B_{out}$, find a point that lies in $P \cap B_{in} \setminus B_{out}$.
With such a data structure,
it is possible to reduce an AIFP query with constraint $(r_{B}, d_{min}, d_{max})$
to a series of BD queries.
Below we briefly describe the idea.
Let $r_0 = d_{min}$,
and for $i=1,2,3 \ldots$, let $r_{i} = (1+\epsilon) r_{i-1}$,
where $\epsilon > 0$ is an approximation factor.
For $i = 0,1,\ldots$, we try to determine whether there is a point in $P \cap B$ whose distance to $q$ is larger than
$r_i$. Note that this can be achieved by a BD query with $B_{in} := B$ and $B_{out}$ being the ball centered at $q$ and with radius
$r_i$.
By iteratively doing this, eventually we will reach an index $j$ such that it is possible to find a point $p \in B \cap P$ that satisfies the condition of $\norm{p - q} > r_{j}$, but no point lies in $P \cap B$ whose distance to $q$ is larger than
$r_{j+1} = (1+\epsilon)r_j$. Thus, $p$ is a $(1-O(\epsilon))$-approximate farthest point to $q$ in $P \cap B$.
From the definition of constrained AIFP query,
it is not hard to see that this process finds the AIFP after at most $O(\log_{1+\epsilon} \frac{d_{max}}{d_{min}})$ iterations.
Every BD data structure supports only $B_{in}$ and $B_{out}$ with fixed radii.
This means that we need to build $O(\log_{1+\epsilon} \frac{d_{max}}{d_{min}})$ BD data structures for answering
any AIFP query with constraint $(r_{B}, d_{min}, d_{max})$.

With the constrained AIFP data structure, we then extend it to a data structure for answering general AIFP queries. 
Our main idea is to use the aforementioned multi-scale construction technique to build a collection of constrained data structures,
which can effectively cover (almost) all possible cases of the radius of $B$ and the farthest distance from $q$ to 
any point in $B \cap P$.
More specifically, for any AIFP query, it is always possible to 
either answer the query easily without using any constrained data structures,
or find a constrained data structure such that the AIFP query satisfies the constraint $(r_{B}, d_{min}, d_{max})$, 
and thus can be used to answer the AIFP query.   

For AMEB query,
we follow the main idea of Badoiu and Clarkson \cite{badoiu2003-mebcore},
and show that an AMEB query is reducible to a sequence of AIFP queries.
More discussions are left to Section \ref{sec-ameb}.

\section{Constrained AIFP Query}

In this section, we discuss how to construct a data structure to answer constrained AIFP queries. 
Particularly, 
given any ball $B$ and point $q$ satisfying the constraint $(r_{B}, d_{min}, d_{max})$, 
\begin{itemize}
    \item the radius of $B$ is $r_{B}$,
    \item the distance from $q$ to its farthest point to $P \cap B$ is within the range of $( d_{min}, d_{max} ]$,
\end{itemize}
the data structure can find the AIFP to $q$ in $P \cap B$ in sub-linear time (with high probability).

In the following, we let $\epsilon > 0$ be an approximation factor, $\gamma > 0$ be a factor that controls the
region fuzziness and $0 < \delta < 1$ be a factor controlling the query success probability.
The main result of this section is summarized as the following lemma.

\begin{lemma}
\label{lm-aifp}
Let $P$ be a set of $n$ points in $\mathbb{R}^d$. 
It is possible to build a
data structure for $P$ with size $O_{\epsilon,\gamma}(dn^{1+\rho} \log \delta^{-1} \log (d_{max}/d_{min}))$ 
in $O_{\epsilon,\gamma}(dn^{1+\rho} \log \delta^{-1} \log (d_{max}/d_{min}))$ time, 
where $0 < \rho < 1$ is a real number depending on $\epsilon$ and $\gamma$,
and the constants hidden in the big-O notation
depend only on $\epsilon, \gamma$.
Given any query $(B,q)$ that satisfies the constraint of  $(r_{B}, d_{min}, d_{max})$,
with probability at least $1 - \delta$,
the data structure finds an $(\epsilon,\gamma)$-AIFP for $q$ in $P \cap B$ within time
$O_{\epsilon,\gamma}(dn^{\rho} \log \delta^{-1} \log (d_{max}/d_{min}))$.
\end{lemma}

In the following, we consider an AIFP query that satisfies constraint $(r_{B}, d_{min}, d_{max})$.
As mentioned in last section,
it is possible to reduce a constrained AIFP query to a series of
\emph{ball difference(BD) range queries},
which report a point in $P$ that lies (approximately) in $B_{in} \setminus B_{out}$ for a given
pair of $\mathbb{R}^d$ balls $(B_{in}$ and $B_{out})$,
or return NULL if no such point exists.
Below, we describe the reduction using a ball-peeling strategy.
We consider a series of balls $B_0, B_{1}, B_{2} ,\ldots$ concentric at $q$
with an exponentially increasing radius.
Let $\xi > 0$ be a to-be-determined approximation factor, and $B_0 := \mathcal{B}(q, d_{min})$ which is the ball centered at $q$ with radius $d_{min}$. 
For integer $i > 0$, let $B_{i+1} = B_{i}(1+\xi)$ which is the ball obtained by enlarging the radius of $B_{i}$ 
by a factor of $(1+\xi)$.
\footnote{Throughout this paper we use similar notations. Let $q$ be any point and $x>0$ be real number.
Then, $\mathcal{B}(q,x)$ denotes the ball centered at $q$ and with radius $x$. Let $B$ be any ball. For real number $y>0$, we let $B(y)$ denote the ball obtained by enlarging (or shrinking if $y < 1$) the radius of $B$ by a factor of $y$.}
For $i=0,1,2,\ldots$, repeatedly perform a BD query with $B_{in} := B$ and $B_{out} := B_{i}$,
until an index $j$ is encountered such that
the BD query reports a point $p_j$ that lies in $P \cap B \setminus B_{j}$, but returns NULL when trying to find a point in $P \cap B \setminus B_{j+1}$.
If $\xi$ is a small enough constant, it is not hard to see that $p_j$ is a good approximation of the IFP to $q$ in $P \cap B$.
Note that in this process, no more than
$\log_{1+\xi} (d_{max}/d_{min})$ BD queries are required.
This is because the distance between $q$ and any point in $B$ is at most $d_{max}$.
Thus, it is not necessary to increase the radius of $B_{out}$ to be more than $d_{max}$ in the BD range query. 
The bound on the number of BD range queries then follows from the facts that the series of BD range queries starts with a $B_{out}$ ball of radius $d_{min}$ 
and each time the radius of $B_{out}$ is increased by a factor of $1+\xi$. 
This process is similar to peel a constant portion of $B_{in}$ each time by $B_{out}$. 
See Figure \ref{fig-ifp} for an illustration.

\begin{figure}[htbp]
\centering
\includegraphics[width=0.4\textwidth]{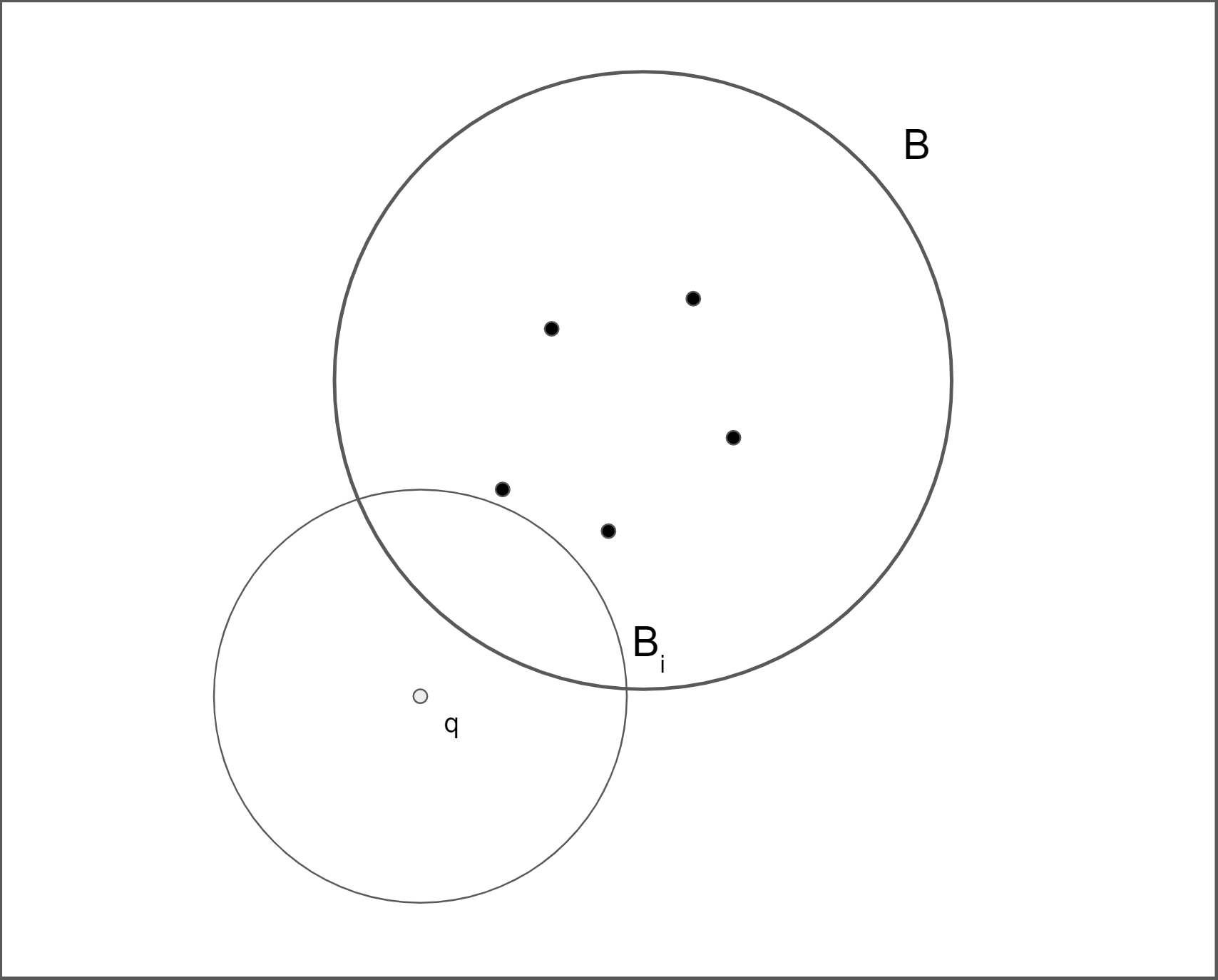}
\includegraphics[width=0.4\textwidth]{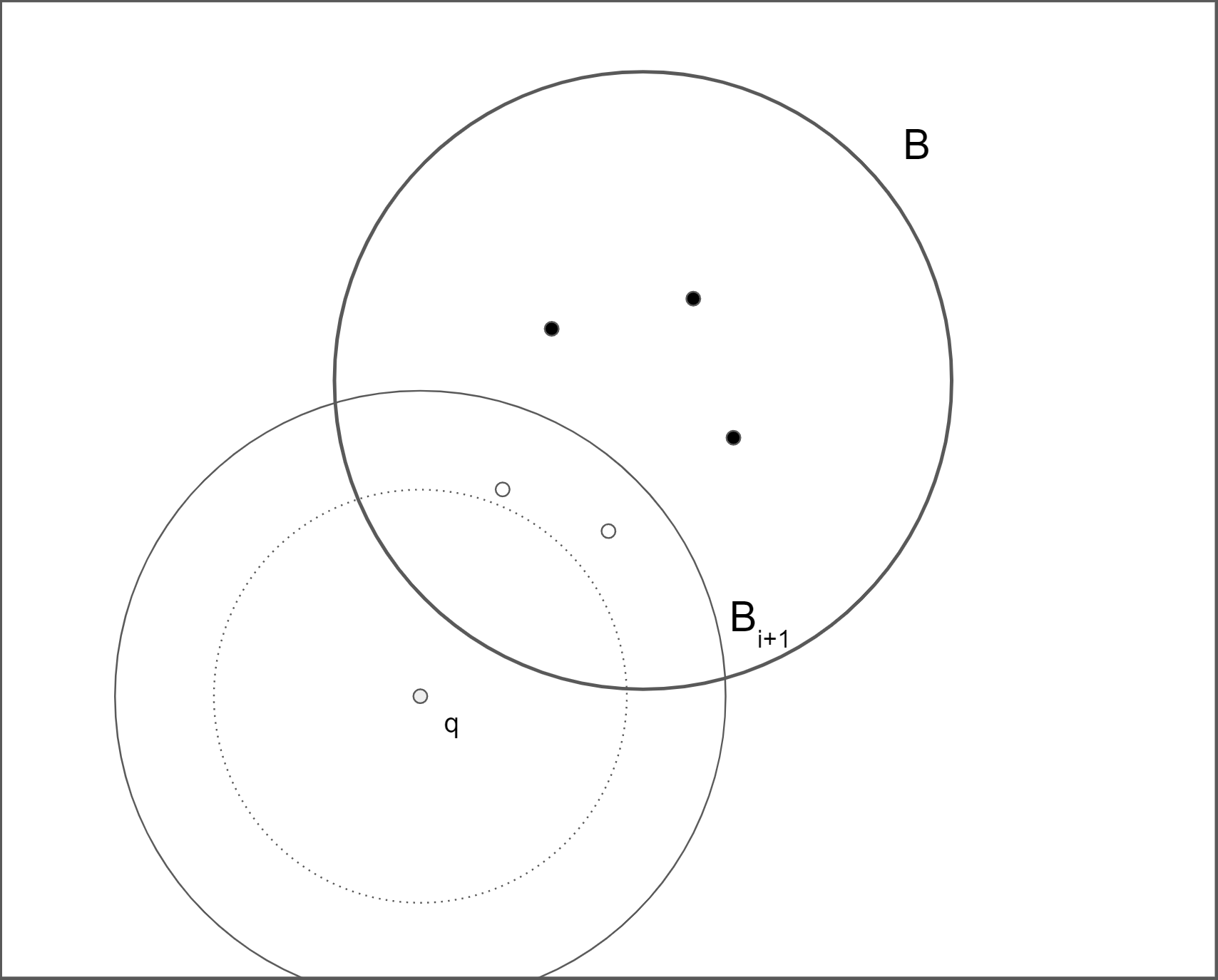}
\includegraphics[width=0.4\textwidth]{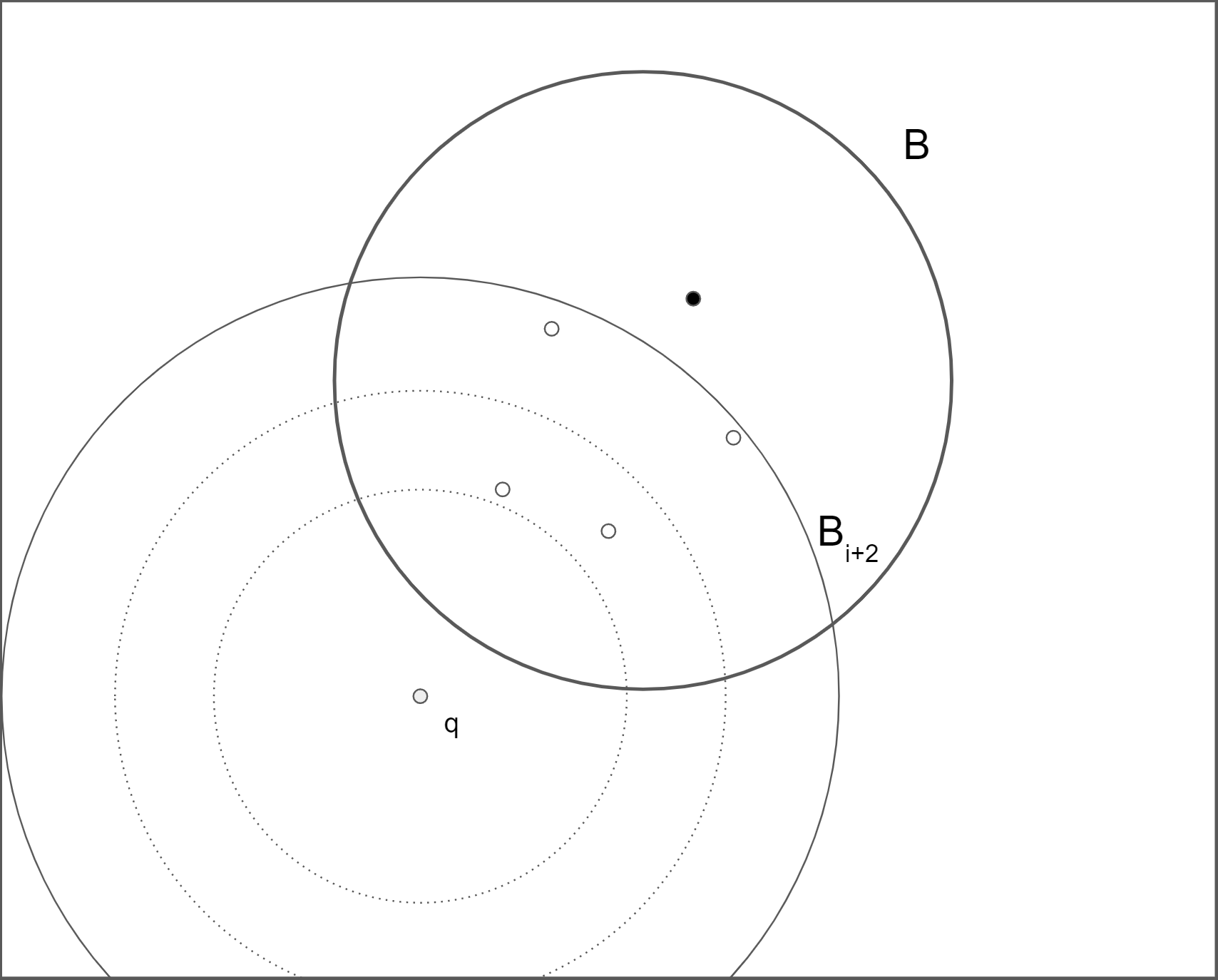}
\includegraphics[width=0.4\textwidth]{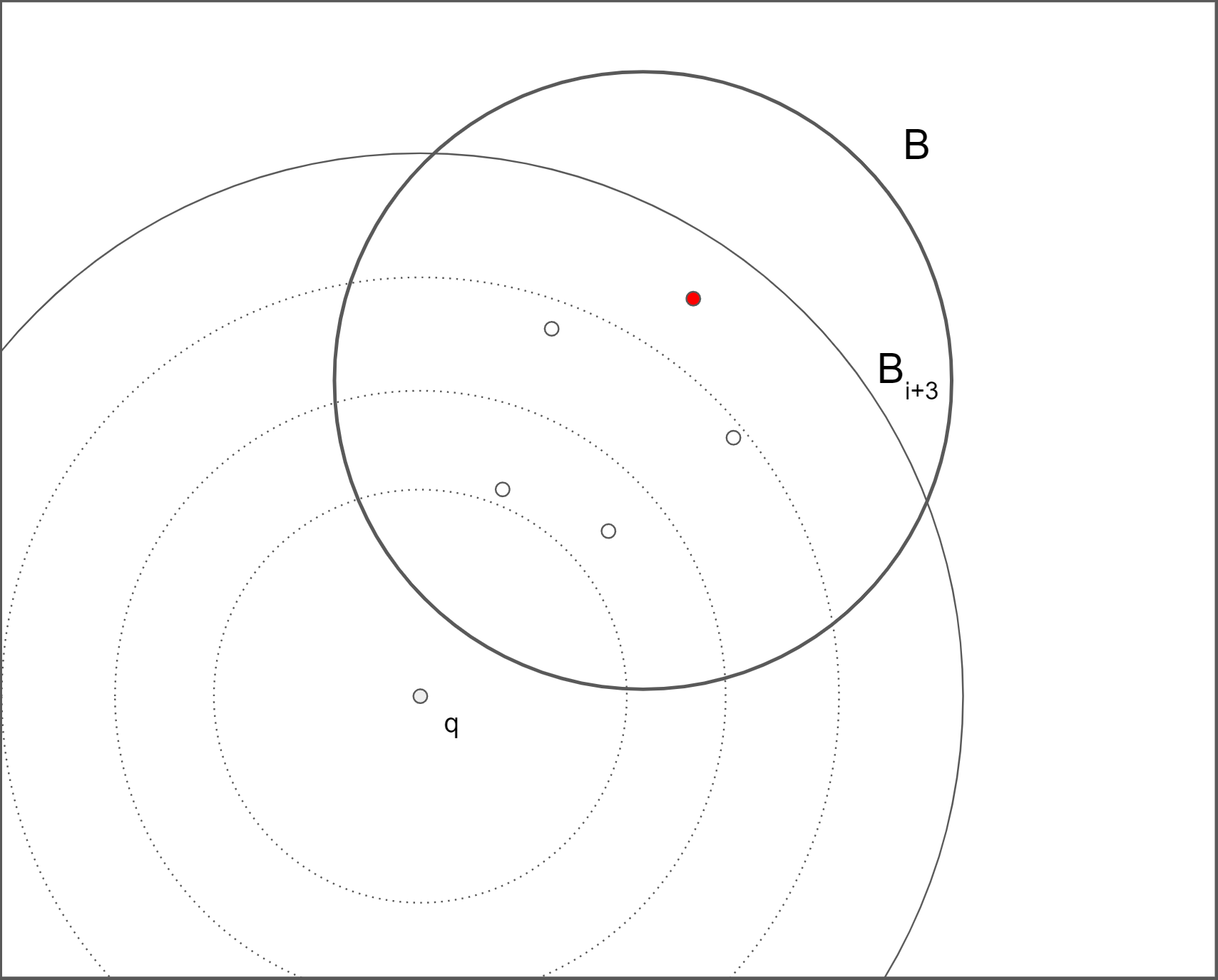}
\caption{An illustration of answering a constrained AIFP query using BD queries.
}
\label{fig-ifp}
\end{figure}

The above discussion suggests that 
a constrained AIFP data structure can be built through 
(approximate) BD query data structures, which have the following definition. 
Let $\xi > 0$ be an approximation factor. 
A data structure is called $\xi$-error BD for a point set $P$,
if given any balls 
$B_{in}$ and $B_{out}$,
it answers the following query (with high success probability):

\begin{enumerate}
\item If there exists a point in $P \cap (B_{in} \setminus B_{out})$, the data structure returns a
point in $P \cap (B_{in}(1+\xi) \setminus B_{out}((1+\xi)^{-1}))$.
\item Otherwise, it returns a
point in $P \cap (B_{in}(1+\xi) \setminus B_{out}((1+\xi)^{-1}))$
or  NULL.
\end{enumerate}

The details of how to construct a  $\xi$-error BD data structure is left to
the next subsection.
Below is 
the main result of the BD query data structure for 
$\xi > 0$, fixed constant $r_{in} > 0, r_{out} > 0$ and success probability controlling factor $0 < \delta < 1$. 

\begin{restatable}{lemma}{BD}
\label{lm-bd}
It is possible to build a $\xi$-error BD query data structure of size  $O_{\xi}(dn^{1+\rho} \log \delta^{-1})$ in
$O_{\xi}(dn^{1+\rho} \log \delta^{-1})$ time, where $0 < \rho < 1$ depends only on $\xi$. The query time of this data structure is 
$O_{\xi}(dn^{\rho} \log \delta^{-1})$.
For any pair of query balls $B_{in}$ and $B_{out}$ with radius $r_{in}$ and $r_{out}$, respectively,
the data structure answers the query with success probability at least $1 - \delta$.
\end{restatable}

Note that each BD query data structure works only for
query balls $B_{in}$ and $B_{out}$ with fixed radii 
$r_{in}$ and $r_{out}$, respectively.
This means that the constrained AIFP data structure should consist of multiple BD data structures with
different values of $r_{in}$ and $r_{out}$.

From the above discussion, we know that a constrained AIFP data structure can be built by
constructing a sequence of $\log(d_{max}/d_{min})$ BD data sturctures 
with $r_{in} := r_B$ and $r_{out}$ being  $d_{min}, (1+\xi) d_{min}, (1+\xi)^2 d_{min}, 
\ldots$.
Such a data structure will allow us to answer constrained AIFP queries using the ball peeling strategy.

Given any constants $\epsilon > 0$, $\gamma > 0$, $0 < \delta < 1$,
and constraint $(r_B,d_{min}, d_{max})$,
the following Algorithm \ref{alg-build-caifp} builds a constrained AIFP
data structure for a given point set $P$.
The data structure is simply a collection of BD query data structures.

\begin{algorithm}[th]
\caption{Build-CAIFP$(P;\epsilon,\gamma,\delta; r_{B}, d_{min}, d_{max})$}

\textbf{Input:} A $\mathbb{R}^d$ point set $P$ with cardinality $n$. 
Constants $\epsilon > 0, \gamma > 0, 0 < \delta < 1$.
Constraint tuple $(r_{B}, d_{min}, d_{max})$. \\
\textbf{Output:} A number of BD-Query data structures built with different parameters.
\label{alg-build-caifp}

\begin{algorithmic}[1]

\State{ Let $\xi = \min \{ (1-\epsilon)^{-1/2} - 1, \gamma \}$.
Construct a sequence of real numbers $r_0,r_1,r_2,\ldots,r_m$, by letting $r_0 = d_{min}$, $m$ be the integer such that $r_0 (1+\xi)^{m-1} < d_{max}$ and
$r_0 (1+\xi)^{m} \geq d_{max}$, $r_{i} = (1+\xi) r_{i-1}$ for $i = 1,2,\ldots,m$, and $\delta' = \delta/m$. 
}

\State{ \textbf{FOR}  $i = 0, 1,2,\ldots,m$, build a $\xi$-error BD query data structure for query balls with radii
$r_{in} = r_{B}$ and $r_{out} = r_i$, with query success probability at least $1-\delta'$.
}

\end{algorithmic}
\end{algorithm}

With such a collection of BD query data structures, we can answer any constrained AIFP query satisfying $(r_{B}, d_{min}, d_{max})$
by applying the ball peeling strategy mentioned before.
The algorithm is formally described as the Algorithm \ref{alg-query-caifp} below.

\begin{algorithm}[th]
\caption{Query-CAIFP$(B,q)$}

\textbf{Input:} A constrained AIFP query $(B,q)$ with constraint $(r_{B}, d_{min}, d_{max})$. \\
\textbf{Output:} A point $p_{ans}$ that is an approximate farthest point in
$B \cap Q$ to $p$, or NULL if no such point exists.
\label{alg-query-caifp}

\begin{algorithmic}[1]

\State{ Initialize variable $p_{ans}$ $\leftarrow$ NULL.
}

\noindent \textbf{Note:} In the following, we use $m$ and $r_i$ for $i=0,1,\ldots,m$ as in 
Algorithm \ref{alg-build-caifp}.

\State{ \textbf{For} $i$ from $0$ to $m$: Make a query $(B,B_{out,i})$ to the BD-Query data structure $BD_i$, by letting 
$B_{out,i} := \mathcal{B}(q, r_i)$.
If the query answer is NULL, \textbf{Return} $p_{ans}$.
Otherwise update $p_{ans}$ to be the query answer.
}

\State{\textbf{Return} $p_{ans}$.
}

\end{algorithmic}
\end{algorithm}

By some simple calculation, we know that the probability that 
all the BD queries in Algorithm \ref{alg-query-caifp} are successful is at least $1-\delta$,
and when this happens, the output point $p_{ans}$ is an $(\epsilon,\gamma)$-AIFP of $q$ in $B \cap P$.
This is summarized as the following lemma.

\begin{lemma}
\label{lm-aifp-correctness}
With probability at least $1-\delta$, Algorithm \ref{alg-query-caifp} outputs a 
point $p_{ans} \in B(1+\gamma)$ such that for any $q \in B \cap P$, 
$\norm{p_{ans} - p} \geq (1-\epsilon) \norm{q - p}$.
\end{lemma}

Next we analyze the space/time complexity of the AIFP scheme.
The query data structure is a combination of
$m = O_{\epsilon,\gamma}(\log (d_{max}/d_{min}))$ BD data structures.
From the discussion of BD data structures (see Lemma \ref{lm-bd}),
every BD query data structure we build has space/time complexity $O_{\epsilon,\gamma}(dn^{1+\rho} \log \delta^{-1})$
where $0<\rho<1$ depends only on $\epsilon,\gamma$. Each BD query takes $O_{\epsilon,\gamma}(dn^{\rho} \log \delta^{-1})$ time.
Lemma \ref{lm-aifp} then follows.

\subsection{The BD Query Scheme}
\label{sec-bd}

In this subsection we present the BD query scheme.
To our best knowledge, this is the first theoretical result to consider the BD range search problem. 
A very special case of BD query called the ``annulus queries'' where the two balls are co-centered is studied in \cite{aumuller2018distance}.
Nonetheless, the technique is not directly applicable to general BD queries.
Our BD range query scheme is based on the classic Locality Sensitive Hashing (LSH) technique \cite{har2012approximate,andoni2014beyond,datar2004locality}
which has been a somewhat standard technique for solving the proximity problems in high dimensional space. 
The main idea of LSH is to utilize a family of hash functions (called an LSH family)
 that have some interesting properties.
Given two points $p$ and $q$ in $\mathbb{R}^d$,
if we randomly pick a function $h$ from the LSH family,
the probability that the event of $h(p) = h(q)$ happens
will be high if $\norm{p-q}$ is smaller than a threshold value, and the probability for the same event will be lower if $\norm{p-q}$ is larger.
Such a property of the LSH family allows us to develop hashing and bucketing based schemes to solve similarity search problems in high dimensional space.
Below is 
the definition of an LSH family.

\begin{definition}
Let $0<r_1<r_2$ and $1 > \mathsf{P}_1 > \mathsf{P}_2 >0$ be any real numbers.
A family $\mathcal{H} = \{ h : \mathbb{R}^d \rightarrow U \}$, where $U$ can be any set of
objects, is called $(r_1,r_2,\mathsf{P}_1,\mathsf{P}_2)$-sensitive, if for any $p,q \in \mathbb{R}^d$.
\begin{enumerate}
    \item if $\norm{p-q} \leq r_1$, then $\text{Pr}_{\mathcal{H}}[h(p) = h(q)] \geq \mathsf{P}_1$,
    \item if $\norm{p-q} > r_2$, then $\text{Pr}_{\mathcal{H}}[h(p) = h(q)] \leq \mathsf{P}_2$.
\end{enumerate}
\end{definition}

It was shown in \cite{datar2004locality} that for any dimension $d$ and any $r>0, c>1$,
an $(r,cr,\mathsf{P}_1,\mathsf{P}_2)$-sensitive family $\mathcal{H}$ exists,
where $1 > \mathsf{P}_1 > \mathsf{P}_2 > 0$ depends only on $c$.
Every hash function $h(p) : \mathbb{R}^d \rightarrow \mathbb{Z}$ maps a point $p$ in $\mathbb{R}^d$ to an integer,
and $h(p)$ has the form $h(p) = \lfloor \frac{a \cdot p + b}{r} \rfloor$
for some $\mathbb{R}^d$ vector $a$ and integers $b,r$.
It takes $O(d)$ time to sample a hash function $h$ from such a family and compute $h(p)$.
Our data structure will make use of two such families.
Let $\mathcal{H}_{in}$ be an $(r_{in},(1+\xi)r_{in},\mathsf{P}_1, \mathsf{P}_2)$-sensitive family, 
and $\mathcal{H}_{out}$ be a $((1+\xi)^{-1}r_{out},r_{out},\mathsf{P}_1$, $\mathsf{P}_2)$-sensitive family,
where \textbf{$0 < \mathsf{P}_1, \mathsf{P}_2 < 1$ are constants depending only on $\xi$}, 
as described in  \cite{datar2004locality}.
Given any BD-query $(B_{in}, B_{out})$ with the centers of the balls being $o_{in}, o_{out}$ respectively,
the family $\mathcal{H}_{in}$ helps us to identify points that are close enough to $o_{in}$
(and therefore lie in $B_{in}$),
and $\mathcal{H}_{out}$ helps us to identify points that are far away enough from $o_{out}$
(and therefore lie outside of $B_{out}$).

\noindent\textbf{High level idea:} Our approach is based on a novel bucketing and query scheme that utilizes the properties of the LSH family.
Before presenting the technical details, We first illustrate the high level idea.
For convenience, we assume for now that the functions in $\mathcal{H}_{in}$ and $\mathcal{H}_{out}$ have range $\{0,1\}$ 
(this is achievable by some simple modification to these hash function families).
We use a randomized process to create a hybrid random hash function $S(\cdot)$ that 
maps any point in $\mathbb{R}^d$ to a bit string.
Such a function $S(\cdot)$ is a concatenation of a number of hash functions
drawn from $\mathcal{H}_{in}$ and $\mathcal{H}_{out}$.
Given $p \in \mathbb{R}^d$, $S(\cdot)$ applies the aforementioned hash functions (drawn from 
$\mathcal{H}_{in}$ and $\mathcal{H}_{out}$) on $p$ to obtain a bit-string.
With such a function $S(\cdot)$, 
consider comparing the bit-strings of $S(p), S(q)$ for points $p,q \in \mathbb{R}$.
Intuitively, based on the properties of $\mathcal{H}_{in}$ and $\mathcal{H}_{out}$, we know that 
if $p,q$ are close enough, $S(p)$ and $S(q)$ should have many common bits in positions
that are determined by functions from $\mathcal{H}_{in}$.
Contrarily, if $p,q$ are far away,
$S(p)$ and $S(q)$ should have only a few common bits in positions
that are determined by functions from $\mathcal{H}_{out}$.

For every point $p \in P$, we use $S(p)$ to compute a bit-string label for $p$,
and put $p$ into the corresponding buckets ({\em i.e.,} labeled with the same bit-strings).
To answer a given BD query $B_{in},B_{out}$ with centers of the balls being $o_{in},o_{out}$, respectively,
we compute $S(o_{in})$ and $S(o_{out})$. 
Note that, based on the above discussion,
we know that if a point $p$ satisfies the condition of $p \in B_{in} \setminus B_{out}$,
then $S(p)$ and $S(o_{in})$ should have many common bits in the positions determined by $\mathcal{H}_{in}$,
and $S(p)$ and $S(o_{out})$ should have few common bits in the positions determined by $\mathcal{H}_{out}$.
Thus, by counting the number of common bits in the labels, we can then locate buckets that are likely to contain points close to $o_{in}$ and far away from $o_{out}$, \emph{i.e.,} points are likely to be in $B_{in} \setminus B_{out}$.
To achieve the desired outcome, we will create multiple set of buckets using multiple random functions $S(p)$.

\noindent\textbf{Details of the Algorithms:}
After understanding the above general idea,
we now present the data structure and the query algorithm along with the analysis.
Let $\mathsf{P}'_1=(1+\mathsf{P}_1)/2,\mathsf{P}'_2=(1+\mathsf{P}_2)/2,\eta = (\mathsf{P}'_1-\mathsf{P}'_2)/3$,  
$a= \lceil (2 \mathsf{P}'_1 \ln 3)/\eta^2 \rceil$.
$\mathsf{P}''_1 = 2^{-2a} \cdot 4/9$, $\mathsf{P}''_2 = 2^{-2a}/3, b = \lceil \log_{1/\mathsf{P}''_2} n \rceil$, $\rho = \frac{\ln 1/\mathsf{P}''_1}{\ln 1/\mathsf{P}''_2}$,
and $c = \lceil n^\rho / \mathsf{P}''_1 \rceil$.
Let $F_{Z}$ be a function that maps
every element in $\mathbb{Z}$ randomly to $0$ or $1$, each with probability $1/2$.
The following Algorithm \ref{alg-buildBD}
shows how to construct a  $\xi$-error BD range query data structure for 
any point set $P$ and radii $r_{in}$ and $r_{out}$.
The data structure consists of $c$ groups of buckets, each
created using a random function $S(p)$ that maps a point to a bit-string of total length $2ab$.

\begin{algorithm}[th]
\caption{CreateBuckets$(P, \xi, r_{in}, r_{out})$}

\textbf{Input:} A point set $P$. 
Parameters $\xi > 0, ,r_{in} > 0, r_{out} > 0$.\\
\textbf{Output:} $\mathcal{G}_1, \mathcal{G}_2,\ldots \mathcal{G}_c$.
Each $\mathcal{G}_i$ is a collection of \emph{buckets} (\emph{i.e.,} sets of points of $P$).
Each bucket $G \in \mathcal{G}_i$ is labeled with \emph{a bit-string}, which is a concatenation of sub-bit-strings
$\textsc{lab}_{in}(G,1)$, $\textsc{lab}_{out}(G,1)$, $\textsc{lab}_{in}(G,2)$, $\textsc{lab}_{out}(G,2)$, $\ldots$,
$\textsc{lab}_{in}(G,b)$, $\textsc{lab}_{out}(G,b)$.
For every $p \in P$ and $i=1,2,\ldots,c$, $p$ appears in one of the buckets in
$\mathcal{G}_i$.

\label{alg-buildBD}

\begin{algorithmic}[1]

\State{ Initialize $\mathcal{G}_i, i=1,2,\ldots,c$, as empty sets.
Each $\mathcal{G}_i$ will be used as a container for buckets.
}

\State{ Randomly sample $abc$ functions from family $\mathcal{H}_{in}$,
and also $abc$ functions from family $\mathcal{H}_{out}$.
Denote these functions as $h_{in,i,j,k}$ and $h_{out,i,j,k}$,
for integers $1 \leq i \leq a, 1 \leq j \leq b, 1 \leq k \leq c$.
For every $h_{in,i,j,k}$, $h_{out,i,j,k}$ and every $p \in Q$,
compute $F_Z(h_{in,i,j,k}(p))$ and $F_Z(h_{out,i,j,k}(p))$.
}

\State{ \textbf{FOR} $k$ from 1 to $c$:
\begin{itemize}
\item For every point $p \in P$, we create a bit-string $S(p)$ that concatenates
$\textsc{lab}_{in}(p,1),\textsc{lab}_{out}(p,1),\textsc{lab}_{in}(p,2),\textsc{lab}_{out}(p,2),\ldots,
\textsc{lab}_{in}(p,b),\textsc{lab}_{out}(p,b)$:
For $j$ from $1$ to $b$, let $\textsc{lab}_{in}(p,j)$, $\textsc{lab}_{out}(p,j)$ be a pair of 
bit-strings of length $a$,
each with the $i$-th bit being $F_Z(h_{in,i,j,k}(p))$, $F_Z(h_{out,i,j,k}(p))$, respectively,
for $i=1,2,\ldots a$.
\item \textbf{IF} there is already a bucket $G$ in $\mathcal{G}_k$ with
label $S(G) = S(p)$, \textbf{DO}:
Put $p$ into $G$.
\item \textbf{ELSE, DO}: Create a new bucket $G$ and put $G$ into $\mathcal{G}_k$,
set the label of $G$ as $S(G) = S(p)$. Put $p$ into $G$.
\end{itemize}
}

\end{algorithmic}
\end{algorithm}

With the BD range query data structure created by the Algorithm \ref{alg-buildBD}, we can use 
the Algorithm \ref{alg-queryBD} below to answer a BD range query for any given pair of 
balls $(B_{in}$ and $B_{out})$.
The main idea of the algorithm compute a bit-string label $S$ for the query,
then examine points in buckets with labels that satisfy certain properties
(\emph{e.g.} should have enough common bits with $S$).
Due to the fact that we label these buckets using functions from two LSH families, it can be shown that the chance for us to find
a point in $B_{in}(1+\xi)  \setminus B_{out}((1+\xi)^{-1})$ 
from one of the examined buckets
will be high if there exists a point in $B_{in} \setminus B_{out}$.

\begin{algorithm}[th]
\caption{BD-Query$(B_{in}, B_{out})$}

\textbf{Input:} Two $\mathbb{R}^d$ balls: $B_{in}$ with center $o_{in}$ and radius $r_{in}$,
and $B_{out}$ with center $o_{out}$ and radius $r_{out}$. 
Assume that collections $\mathcal{G}_1, \mathcal{G}_2,\ldots \mathcal{G}_c$
 have already been generated by algorithm CreateBucket. \\
\textbf{Output:} A point $p \in Q$, or NULL. \\
Note: The algorithm  probes a number of points in the buckets of $\mathcal{G}_1, \mathcal{G}_2,\ldots \mathcal{G}_c$ until
a suitable point is found as the output, or terminates and returns NULL when no such point can be found or the number of probes exceeds a certain limit.
\label{alg-queryBD}

\begin{algorithmic}[1]

\State{ Do the following, but terminate and return NULL when $3c$ points are examined: 
\textbf{FOR} $k$ from $1$ to $c$:
\begin{itemize}
\item Create a bit string $S$ that concatenates
$\textsc{lab}_{in}(o_{in},1)$, $\textsc{lab}_{out}(o_{out},1)$, $\textsc{lab}_{in}(o_{in},2)$, $\textsc{lab}_{out}(o_{out},2)$, $\ldots$, 
$\textsc{lab}_{in}(o_{in},b)$, $\textsc{lab}_{out}(o_{out},b)$:
For $j$ from $1$ to $b$, 
let $\textsc{lab}_{in}(o_{in},j)$, $\textsc{lab}_{out}(o_{out},j)$ be a pair of bit strings of length $a$,
with the $i$-th bit of each string being $F_Z(h_{in,i,j,k}(o_{in}))$, $1 - F_Z(h_{out,i,j,k}(o_{out}))$, respectively, for $i=1,2,\ldots a$.
\item Create a bit string $S'$ that concatenates
$\textsc{lab}_{in}'(1),\textsc{lab}_{out}'(1),\textsc{lab}_{in}'(2),\textsc{lab}_{out}'(2),\ldots,$
$\textsc{lab}_{in}'(b),\textsc{lab}_{out}'(b)$:
Each of these sub bit strings is a random bit string of length $a$, drawn uniformly randomly from $\{ 0,1 \}^a$.
\item If there exists some integer $j$ such that
$\text{COM}(\textsc{lab}_{in}(o_{in},j), \textsc{lab}_{in}'(j)) < t_1$ or
$\text{COM}(\textsc{lab}_{out}(o_{out},j), \textsc{lab}_{out}'(j)) < t_2$,
where $\text{COM}(x,y)$ counts the number of common digits of 2 bit strings $x,y$,
$t_1 = \mathsf{P}'_1 a - \eta a, t_2 = (1-\mathsf{P}_2)a/2 - \eta a$
\textbf{CONTINUE.}
\item If there is no bucket in $\mathcal{G}_k$ that is labeled with $S'$,
\textbf{CONTINUE.}
\item Examine all the
points in the bucket $G$ in $\mathcal{G}_k$ that is labeled with $S'$.
Stop when there a point $p \in G$ such that
$p \in B_{in}(1+\xi)  \setminus B_{out}((1+\xi)^{-1})$.
Return $p$.
\end{itemize}
}

\State{Return NULL if no point is returned in the above process.}

\end{algorithmic}
\end{algorithm}

In the following we show the correctness of Algorithm \ref{alg-queryBD}.
Consider the for loop in Step 1 of Algorithm \ref{alg-queryBD} when answering
a query $(B_{in},B_{out})$.
Using the notations from Algorithm \ref{alg-queryBD},
for any $k$ from 1 to $c$ in Step 1, we have the following lemma,
which shows
that if a point in $P$ lies in (or outside of) the query range,
the number of common bits between its bucket label and the label computed from the query
would likely (or unlikely) be high, respectively.

\begin{lemma}
\label{lm-bd-1}
Let $p \in P$ be a point that lies in $B_{in} \setminus B_{out}$, and $q \in P$ be a point that does NOT lie in
$B_{in}(1+\xi)  \setminus B_{out}((1+\xi)^{-1})$.
Let 
$S(p) = \textsc{lab}_{in}(p,1),\textsc{lab}_{out}(p,1),\textsc{lab}_{in}(p,2),\textsc{lab}_{out}(p,2)$, $\ldots,\textsc{lab}_{in}(p,b),\textsc{lab}_{out}(p,b)$ and  
$S(q) = \textsc{lab}_{in}(q,1),\textsc{lab}_{out}(q,1),\textsc{lab}_{in}(q,2),\textsc{lab}_{out}(q,2),\ldots$, $\textsc{lab}_{in}(q,b),\textsc{lab}_{out}(q,b)$ 
be the labels of the bucket
in $\mathcal{G}_k$ that contains $p$ and $q$, respectively.
For any $j=1,2,\ldots,b$, the following holds.
\begin{itemize}
\item $\text{Pr}[\text{COM}(\textsc{lab}_{in}(p,j),\textsc{lab}_{in}(o_{in},j)) \geq t_1 \wedge \text{COM}(\textsc{lab}_{out}(p,j),\textsc{lab}_{out}(o_{out},j)) \geq t_2] \geq 4/9$.
\item $\text{Pr}[\text{COM}(\textsc{lab}_{in}(q,j),\textsc{lab}_{in}(o_{in},j)) \geq t_1 \wedge \text{COM}(\textsc{lab}_{out}(q,j),\textsc{lab}_{out}(o_{out},j)) \geq t_2] \leq 1/3$.
\end{itemize}
\end{lemma}

\begin{proof}
Since $p \in B_{in} \setminus B_{out}$, we have $\norm{p - o_{in}} \leq r_{in}$ and $\norm{p - o_{out}} \geq r_{out}$.
For any hash function $h_1 \in \mathcal{H}_{in}$ and $h_2 \in \mathcal{H}_{out}$,
$\text{Pr}[h_1(p) = h_1(o_{in})] \geq \mathsf{P}_1$ and $\text{Pr}[h_2(p) = h_2(o_{out})] \leq \mathsf{P}_2$.
Thus, we have $\text{Pr}[F_Z(h_1(p)) = F_Z(h_1(o_{in}))] \geq (\mathsf{P}_1+1)/2$ and $\text{Pr}[F_Z(h_2(p)) = 1-F_Z(h_2(o_{out}))] \leq (1-\mathsf{P}_2)/2$.
This means that for any $i=1,2\ldots,a$,
the probability that the $i$-th bit of $\textsc{lab}_{in}(p,j)$ is the same as that of $\textsc{lab}_{in}(o_{in},j)$
is at least $\mathsf{P}'_1 = (\mathsf{P}_1+1)/2$.
Since the hash functions to determine each of the bits are drawn independently, an estimation of $X = \text{COM}(\textsc{lab}_{in}(p,j),\textsc{lab}_{in}(o_{in},j))$
can be obtained by $\text{Pr}[F_Z(h_1(p)) = F_Z(h_1(o_{in}))] \geq (\mathsf{P}_1+1)/2$
using the concentration inequalities for binomial distributions.
Using a variant of the Chernoff inequalities from \cite{mitzenmacher2017probability}, we have
\[
\text{Pr}[X \leq \mathsf{P}'_1 a - \eta a] \leq e^{-(\eta a)^2/(2 \mathsf{P}'_1 a)}.
\]
From the definition of the parameters, 
we know that $\text{Pr}[X \leq \mathsf{P}'_1 a - \eta a] \leq 1/3$ (by simple calculation).
Thus, we have $\text{Pr}[X \geq t_1] \geq 2/3$.

Let $Y = \text{COM}(\textsc{lab}_{out}(p,j),\textsc{lab}_{out}(o_{out},j))$. From $\text{Pr}[F_Z(h_2(p)) = 1-F_Z(h_2(o_{out}))] \leq (1-\mathsf{P}_2)/2$
and using a similar argument as above, we can also obtain $\text{Pr}[Y \geq t_2] \geq 2/3$ (the details are omitted).
Since the hash functions are drawn independently,
we have $\text{Pr}[X \geq t_1 \wedge Y \geq t_2] \geq 4/9$.

In the following we discuss the case that $q \not\in B_{in}(1+\xi) \setminus B_{out}((1+\xi)^{-1})$.
This means either $\norm{q - o_{in}} \geq (1+\xi)r_{in}$ or $\norm{q - o_{out}} \leq (1+\xi)^{-1} r_{out}$.
We first consider the case $\norm{q - o_{in}} \geq (1+\xi)r_{in}$.
For any hash function $h_1 \in \mathcal{H}_{in}$,
$\text{Pr}[h_1(q) = h_1(o_{in})] \leq \mathsf{P}_2$.
Thus, we have $\text{Pr}[F_Z(h_1(q)) = F_Z(h_1(o_{in}))] \leq (\mathsf{P}_2+1)/2$.
This means that for any $i=1,2\ldots,a$,
the probability that the $i$-th bit of $\textsc{lab}_{in}(q,j)$ is  the same as that of $\textsc{lab}_{in}(o_{in},j)$
is at most $\mathsf{P}'_2 = (\mathsf{P}_2+1)/2$.
Again, we use a concentration inequality to obtain an estimation of $X = \text{COM}(\textsc{lab}_{in}(q,j),\textsc{lab}_{in}(o_{in},j))$.
Using a variant of the Chernoff inequalities from \cite{mitzenmacher2017probability}, we have
\[
\text{Pr}[X \geq \mathsf{P}'_2 a + \eta a] \leq e^{-(\eta a)^2/(2 \mathsf{P}'_2 a + \eta a/3)}.
\]
Note that $a = (2\mathsf{P}'_1 \ln 3) / \eta^2 \geq ((2\mathsf{P}'_2 + \eta/3) \ln 3) / \eta^2$,
which implies that $e^{-(\eta a)^2/(2 \mathsf{P}'_2 a + \eta a/3)} \leq 1/3$ (by simple calculation).
Thus, we have $\text{Pr}[X \geq \mathsf{P}'_2 a + \eta a] \leq 1/3$.
Also, since $\mathsf{P}'_2 a + \eta a < \mathsf{P}'_1 a - \eta a = t_1$,
we get $\text{Pr}[X \geq t_1] \leq 1/3$.
This immediately implies that 
$\text{Pr}[\text{COM}(\textsc{lab}_{in}(q,j),\textsc{lab}_{in}(o_{in},j)) \geq t_1 \wedge \text{COM}(\textsc{lab}_{out}(q,j),\textsc{lab}_{out}(o_{out},j)) \geq t_2] \leq 1/3$.

The argument for the case $\norm{q - o_{out}} \leq (1+\xi)^{-1} r_{out}$ is similar. Thus, we omit it here.
This completes the proof. 
\end{proof}

From the above lemma, we can conclude the following by basic calculation.
For any $k$ from 1 to $c$ in Step 1 of Algorithm \ref{alg-queryBD} (if the loop is actually executed), 
let $p \in P$ be a point that lies in $B_{in} \setminus B_{out}$,
and $q \in P$ be a point that does NOT lie in
$B_{in}(1+\xi)  \setminus B_{out}((1+\xi)^{-1})$,
we have the following.

\begin{lemma}
\label{lm-bd-2}
Let $G_p$ and $G_q$ be the buckets in $\mathcal{G}_k$ that contain $p$ and $q$, respectively.
The probability for $G_p$ to be examined is no smaller than $2^{-2ab} (4/9)^b = (\mathsf{P}''_1)^b$,
and the probability for the event ``ALL such $G_p$ for $k$ from 1 to $c$ are NOT examined'' is at most
$(1-(\mathsf{P}''_1)^b)^c \leq 1/e$.
The probability for $G_q$ to be examined is no larger than $2^{-2ab} (1/3)^b = (\mathsf{P}''_2)^b \leq 1/n$.
\end{lemma}

With the above lemma, we can obtain the following lemma using an argument similar to \cite{har2012approximate} for near neighbor search with LSH.
This proves the correctness of the query scheme.

\begin{lemma}
\label{lm-bd-3}
If there exists a point in $P$ that lies in $B_{in} \setminus B_{out}$,
 with probability at least $1/4$,
Algorithm \ref{alg-queryBD} reports a point in $P$
that lies in $B_{in}(1+\xi)  \setminus B_{out}((1+\xi)^{-1})$.
\end{lemma}

The complexity of the data structure is $O(dabcn)$, which is $O_{\xi}(dn^{1+\rho}\log n)$
from $a = O_{\xi}(1), b = O_{\xi}(\log n)$ and $c = O_{\xi}(n^\rho)$, and $\rho$ and the constant hidden in the big-O notation depends only on $\xi$. The query time is $O(abcd)$, which is $O_{\xi}(dn^{\rho}\log n)$.
To achieve $1-\delta$ success probability, it suffices to concatenate $O(\log \delta)$ such data structures together.

We leave the proof of the above 2 lemmas and the full proof of Lemma \ref{lm-bd}
to the appendix for readability.

\section{Multi-scale Construction}

In this section, we present the  {\em multi-scale construction} method, which  is a standalone technique with potential to be used to other high dimensional range-aggregate query problems.

The multi-scale construction method is motivated by
several high dimensional geometric query problems that share the following common feature: they are challenging in 
the general settings, but become more approachable if some key parameters are known in advance. 
The AIFP query problem discussed in this paper is such an example. 
In the previous section, we have shown how to construct an AIFP data structure if we fix the size of the query ball
and know that the farthest distance lies in a given range.

The basic ideas behind multi-scale construction are the follows. Firstly, we know that 
if a problem is solvable when one or more key parameters are fixed, a feasible way to solve the general case of the problem is to first enumerate all possible cases of the problem defined by (the combinations of) the values of the parameters. Then, 
solve each case of the problem,
and finally obtain the solution from that of all the enumerated cases.
The multi-scale construction method follows a similar idea. 
More specifically, to obtain a general AIFP query data structure, the multi-scale construction method builds
a set of constrained AIFP query data structures that cover all possible radii of $B$ and farthest distance value. Secondly, since 
it is impossible to enumerate the infinite number of all possible values for these parameters, our idea is to sample a small set of fixed radii (based on the distribution of the points in $P$) and build constrained AIFP 
data structures only for the set of sampled values. This will certainly introduce errors. However, 
good approximations are achievable by using
a range cover technique.

Below we first briefly introduce two key ingredients of our method, Aggregation Tree and Range Cover, and then show how they can be used to form a multi-scale construction. 

\subsection{Aggregation Tree and Range Cover} 
\label{sec-hst}

In this subsection,
we briefly introduce the two components of the multi-scale construction scheme:
the \emph{aggregation tree} and the \emph{range cover} data structure.
We first introduce aggregation tree, which is used in \cite{huang2019-tdrc} as an ingredient of the range cover data structure.
It is essentially a slight modification of  the Hierarchical Well-Separated Tree (HST) introduced in \cite{har2011geometric}. Below is the definition of an aggregation tree:
(1) Every node $v$ (called {\em aggregation node}) represents a subset $P(v)$ of  $P$, and the root represents $P$;
(2)Every aggregation node $v$ is associated with a representative point $re(v) \in P(v)$ and a size $s(v)$. Let $Dia(P(v))$ denotes
the diameter of $P(v)$, $s(v)$ is a polynomial approximation of $Dia(P(v))$: 
$Dia(P(v)) \leq s(v)$, and $\frac{s(v)}{ Dia(P)}$ is upper-bounded by
a polynomial function $\mathcal{P}_{HST}(n,d) \geq 1$ (called  {\em distortion polynomial});
(3) Every leaf node corresponds to one point in $P$ with size $s(v)=0$, and each point appears in exactly one leaf node;
(4) The two children $v_{1}$ and $v_{2}$ of any internal node $v$ form a partition of $v$ with $\max\{s(v_{1}),$ $s(v_{2})\} < s(v)$; and
(5) For every aggregation node $v$ with parent $v_{p}$,  $\frac{s(v_p)}{r_{out}}$ is bounded by
the distortion polynomial $\mathcal{P}_{HST}(n,d) \geq 1$, where  $r_{out}$ is the minimum distance between points in $P(v)$ and
points in $P \setminus P(v)$.

The  above definition is equivalent to the properties of HST in \cite{har2011geometric}, except that we have an additional distortion requirement (Item 5). See Figure \ref{fig-hst} for example of an aggregate tree.

\begin{figure}[htbp]
\centering
\includegraphics[width=0.6\textwidth]{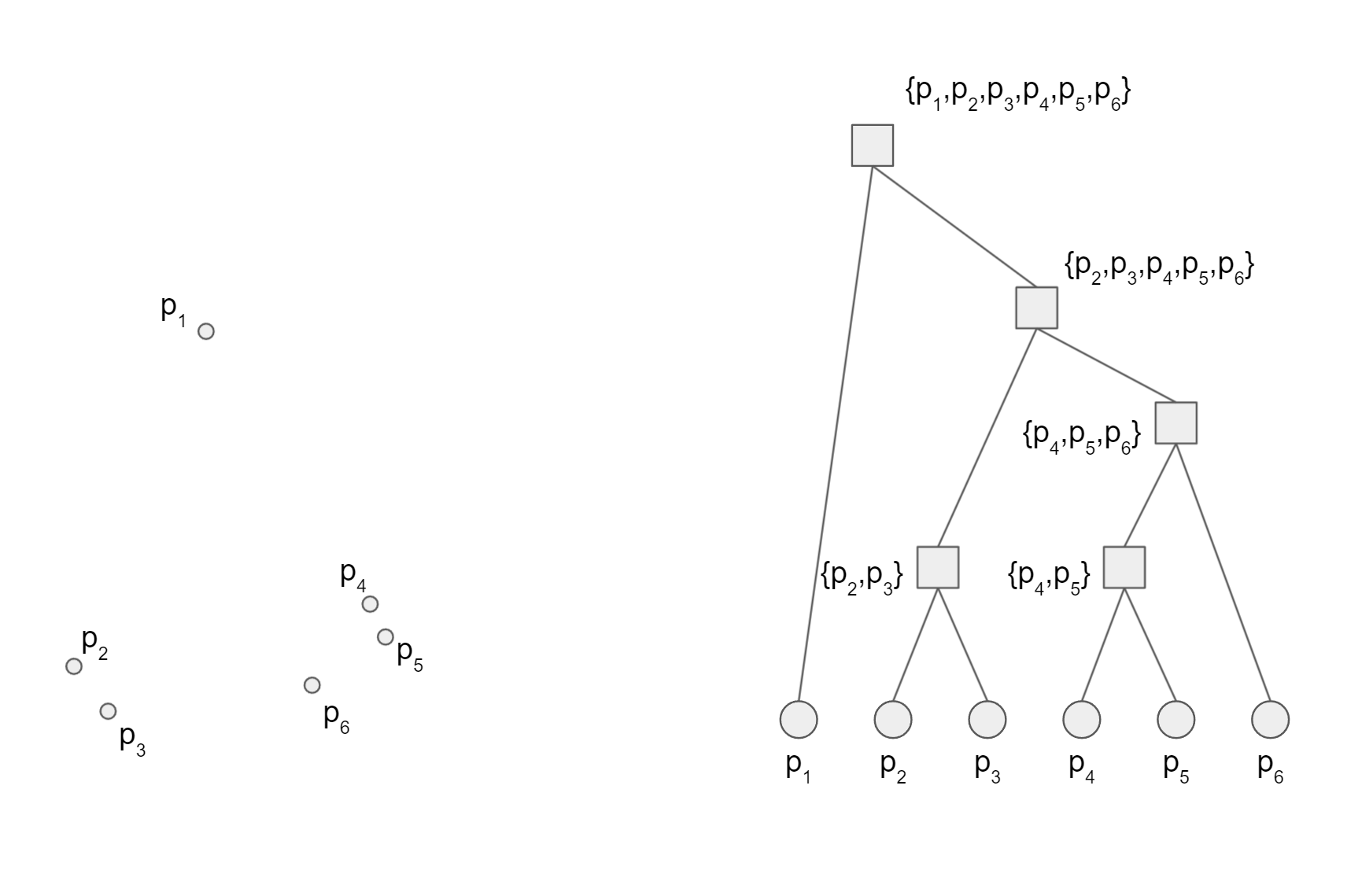}
\caption{An illustration of an aggregation tree built for 6 points.
}
\label{fig-hst}
\end{figure}

An aggregation tree can be constructed in $O(d n \log^2 n)$ time using the method in \cite{har2011geometric}.
It is proved in \cite{har2011geometric} that the distortion polynomial is $\mathcal{P}_{HST}(n,d) = dn$.
In the rest of the paper, we always assume that the distortion of an aggregation
tree is $\mathcal{P}_{HST}(n,d) = dn$.

\begin{algorithm}[th]
\caption{RangeCover$(T_P;\lambda,\Delta)$}

\textbf{Input:} An aggregation tree $T_P$ built over a set $P$ of $n$ points in $\mathbb{R}^d$;
controlling factors $0<\lambda<1$ and an integer $\Delta \geq  4\mathcal{P}_{HST}(n,d)$.\\
\textbf{Output:} A number of buckets, where each bucket stores a  number of tree nodes. Each bucket $B_t$ is indexed by an integer $t$ and associated with an interval
$((1+\lambda)^t,(1+\lambda)^{t+1}]$.
\label{alg-range}

\begin{algorithmic}[1]

\State{ For every integer $t$
create an empty bucket $B_t$ associated with interval $((1+\lambda)^t,(1+\lambda)^{t+1}]$. 
(Note that $B_{t}$ will not be actually created 
until some tree node $v$ is inserted into it.)
}

\State{ For every non-root node $v$ of $T_P$,
let $v_p$ be its parent  in $T_P$,
$r_H$ be $s(v_p)/\lambda$,
and $r_L$ be $\max\{s(v)/\lambda, s(v_p)/\Delta \}$.
Do 
\begin{itemize}
\item For every integer $t$ satisfying inequality
$r_L \leq (1+\lambda)^t < r_H$,
insert $v$ into bucket $B_t$.
\end{itemize}
}
\end{algorithmic}
\end{algorithm}

In the following, we briefly introduce range cover.
Range cover is a technique proposed in \cite{huang2019-tdrc} 
for solving the truth discovery problem 
in high dimensions. We utilize it in a completely different way to form a
multi-scale construction for the AIFP query problem. Below is the algorithm
(Algorithm \ref{alg-range}). 
Given an aggregation tree $T_{p}$ and real number parameters $\Delta \geq 8n$ and 
$0<\lambda <1$
(whose values 
will be determined later), 
the range cover  algorithm creates a number of buckets
for the nodes of $T_P$.
Each bucket $B_t$ is associated with an interval of real number
$((1+\lambda)^t,(1+\lambda)^{t+1}]$. If a value $r$ lies in the interval of a bucket $B_t$,
it can be shown that the diameter of every aggregation node $v$
is small compared to $r$,
and thus all points in $P(v)$ can be approximately viewed as one ``heavy'' point located at the representative point $re(v)$.
Intuitively, every bucket from the range cover algorithm provides a view of $P$ when observed from
a distance $r$ in the range of the bucket,
where each node in the bucket represents a ``heavy'' point
that is formed by the aggregation of a set of close (compared to the observing distance)
clusters of points in $P$.
Thus, the buckets of the range cover provides views of the input point set at different scales of observing distances  (see Figure \ref{fig-range} for an illustration).
The size of the output data structure is only $O(n \log n\Delta)$, as shown in \cite{huang2019-tdrc}.

\begin{figure}[htbp]
\centering
\includegraphics[width=0.6\textwidth]{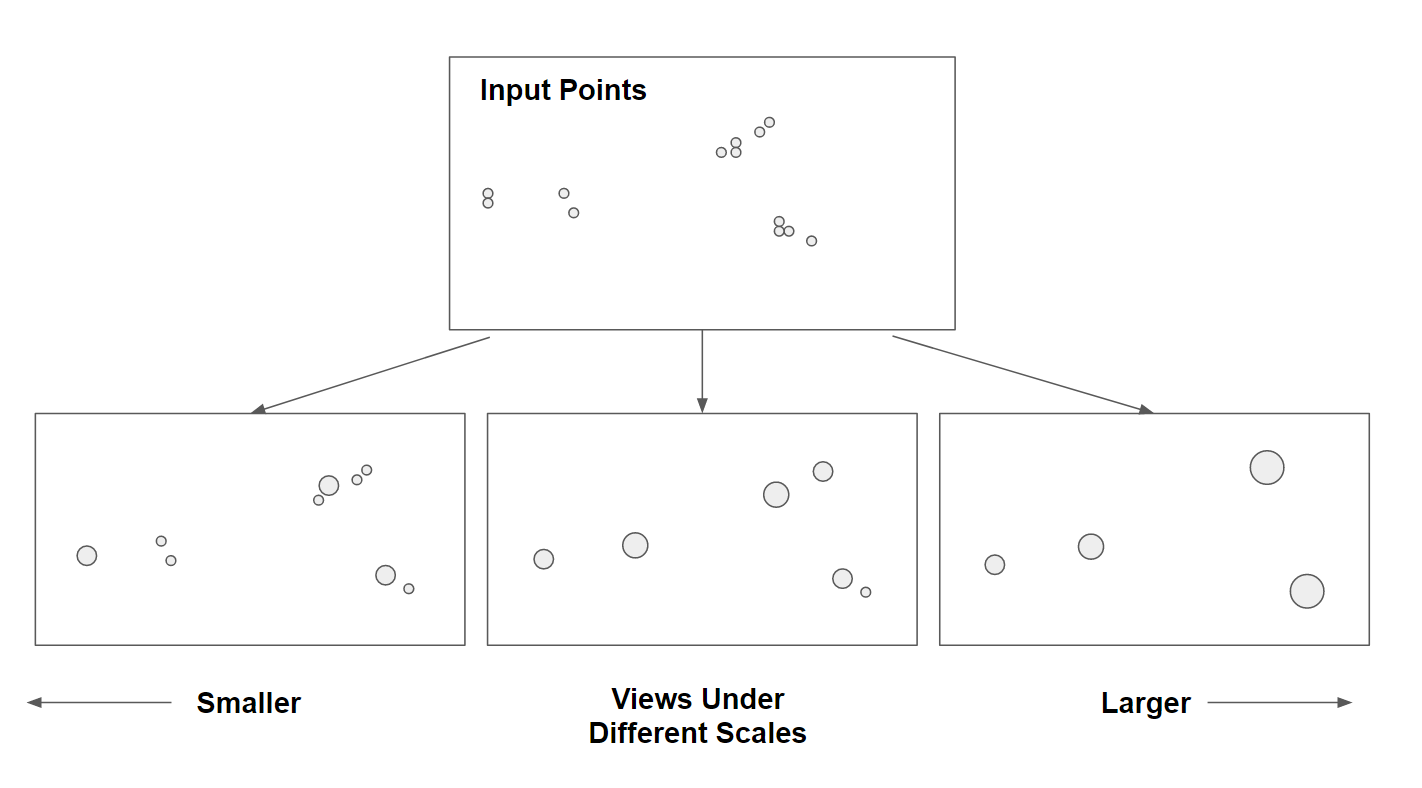}
\caption{An illustration of range cover. The nodes in every bucket can be viewed as ``heavy'' points yielded by the aggregation of  a set of close points. Every bucket provides a view of the input point set when observed from a certain distance. All the buckets jointly form a complete set of views of the input points at all possible scales.}
\label{fig-range}
\end{figure}

Note that for many problems, fixing some key parameters also means fixing the ``observing distance'' of $P$ from the perspective of solving the problem. This allows us to solve the problem based on the view of $P$ provided by the
bucket associated with the corresponding observation distance. We will show that this idea also applies to the AIFP problem.

\subsection{Multi-scale Construction for AIFP}

In this subsection, we use AIFP problem as an example to show how to implement multi-scale construction
using the range cover data structure.

We first observe that every bucket of the range cover can be used to solve
a constrained AIFP problem (with the proof given later).
Given an AIFP query $(B,q)$, 
if the (approximate) distance from $q$ to a point in $P \cap B$  
is known and
falls in the interval of  $((1+\lambda)^t,(1+\lambda)^{t+1}]$,
then the (approximate) distance from $q$ to a point of $B_t \cap B$ (where every node $v$ in $B_t$ is viewed as
a ``heavy'' 
point located at $re(v)$)
is an AIFP of $q$ in $P \cap B$.
This means that $B_t$ provides a good ``sketch'' of $P$ 
that allows more efficient computation of 
the AIFP of $q$ in $P \cap B$.
This observation leads to 
the main idea of the multi-scale construction method. 
To obtain a general AIFP query data structure, 
for every bucket $B_t$, 
we  construct a constrained AIFP data structure for $B_t$ (viewed as a set of ``heavy'' points), exploiting the assumption that the farthest distance is in the interval  of 
$((1+\lambda)^t,(1+\lambda)^{t+1}]$.
To answer a general AIFP query, we can 
find the AIFP for every bucket
by querying the constrained AIFP data structures associated with the bucket.
In this way, we can compute AIFPs for all possible radii.
When answering a general AIFP query, 
we first 
determine an approximate farthest distance of $q$ to $P \cap B$,
and then query the appropriate constrained AIFP data structures.
Despite the necessity of building multiple constrained data structures,
the complexity of the multi-scale construction is not high, 
as
the total number of nodes in all buckets is only $\tilde{O}(n)$.

However, the above idea is hard to implement,
because each bucket $B_t$ is only responsible for
a small range $((1+\lambda)^t,(1+\lambda)^{t+1}]$ of the possible farthest distance from $q$ to $P \cap P$.
This means that we need an accurate estimation of this distance when answering the query, 
which is almost as hard as the query itself.
We resolve this issue by merging multiple consecutive buckets into a larger one.
The resulting bucket can account for a larger range
of the possible farthest distances.
We then build a constrained data structure for each bucket.

This leads to the following Multi-Scale algorithm.
Let $\Gamma \geq 1$ be an integer constant to be determined,
and $\mathcal{A}$ be an algorithm 
for building a constrained data structure.
In this algorithm, for each integer $t$, we try to merge the aggregation nodes
in buckets $B_{t},B_{t+1},\ldots,B_{t+\Gamma}$ from the range cover (recall that
these buckets are associated with farthest distance ranges $((1+\lambda)^t, (1+\lambda)^{t+1}], ((1+\lambda)^{t+1}, (1+\lambda)^{t+2}] \ldots, ((1+\lambda)^{t+\Gamma}, (1+\lambda)^{t+\Gamma+1}]$, respectively)
into one bucket $B^{+}_{t}$ that could account for a larger range $((1+\lambda)^t, (1+\lambda)^{t+\Gamma+1}]$.
We then use $\mathcal{A}$ to build a data structure $\mathcal{S}_t$ for every bucket $B^{+}_{t}$
(by viewing every node in $B^{+}_{t}$ as a point).

\begin{algorithm}[th]
\caption{Multi-Scale$(T_P; \lambda,\Delta, \Gamma; \mathcal{A})$}

\textbf{Input:} An aggregation tree $T_P$ built over a set $P$ of $n$ points in $\mathbb{R}^d$;
controlling factors $0<\lambda<1$, integer $\Delta \geq 4\mathcal{P}_{HST}(n,d) = 4dn$,
and integer $\Gamma \geq 1$. A routine $\mathcal{A}$ which builds a constrained data structure for
any given bucket $B_t$ and point set $re(B^{+}_{t}) := \{ re(v) \mid v \in B^{+}_{t}\}$. \\
\textbf{Output:} A number of buckets, with each storing a number of tree nodes. Each bucket $B^+_t$ is indexed by an integer $t$ and associated with an interval
$((1+\lambda)^t,(1+\lambda)^{t+\Gamma+1}]$.
Each bucket $B^+_t$ is associated with data structure $\mathcal{S}_t$ built by $\mathcal{A}$.

\label{alg-merge}
\label{alg-multi}
\begin{algorithmic}[1]

\State{ Create a collection of buckets $\{ B_t \}$ by calling
RangeCover($T_P;\lambda,\Delta(1+\lambda)^{\Gamma}$).
}

\State{ For each integer $t$
create an empty bucket $B^+_t$ associated with 
$((1+\lambda)^t,(1+\lambda)^{t+\Gamma+1}]$. 
}

\State{ For every non-root node $v$ of $T_P$, 
enumerated in a bottom-up manner in $T_P$ so that the children of a node is always visited earlier than the parent node, 
put $v$ into $B^+_t$ for every $t$ such that the following is satisfied:
\begin{itemize}
\item $s(v) \leq \lambda (1+\lambda)^t$, $v$ appears in $B_{t'} \in \{ B_t \}$ for some $t \leq t' \leq t + \Gamma$, and none of $v$'s descendants are put in $B^+_t$
previously.
\end{itemize}
}

\State{
For every non-empty bucket $B^+_t$, create a data structure $\mathcal{S}_t$ using $\mathcal{A}$ for 
the point set $re(B^{+}_{t}) := \{ re(v) \mid v \in B^{+}_{t}\}$.
}

\end{algorithmic}
\end{algorithm}

For better understanding of this scheme,
we first briefly discuss the geometric properties of 
the buckets created by Algorithm \ref{alg-merge}.
Intuitively speaking, the aggregation nodes of every bucket
provide a sketch of \emph{almost} the whole input point set $P$, with the 
exception being points that satisfying some special isolation property.
This can be briefly described as follows:
(1) The diameter of each the aggregation node (viewed as a point set) should be small to the observation distances;
(2) The aggregation nodes are mutually disjoint; and (3)
Every point $p \in P$ is either in one of the nodes in the bucket, or 
it is in an aggregation node (not in the bucket) whose distance to other nodes is large.
These properties are formalized as follows. 
(We leave the proofs of all the following claims and lemma to the full version of the paper.)

\begin{claim}
\label{cl-range2}
Let $v$ be any aggregation node $v$ in a created bucket $B^+_t$. Then, 
$s(v) \leq \lambda (1+\lambda)^t$.
\end{claim}

\begin{lemma}
\label{lm-range2}
For any $p \in P$ and and bucket $B^+_t$ created by Algorithm \ref{alg-merge},
one of the following holds:
\begin{enumerate}
    \item There exists exactly one aggregation node $v \in B^+_t$ such that $p \in P(v)$.
    \item Either $B^+_t$ is empty or there exists no aggregation node $v \in B^+_t$ such that $p \in P(v)$.
    There exists an aggregation node $v'$ in  $T_P$ such that
    $s(v')/\lambda \leq (1+\lambda)^t$.
    Furthermore,
    let $q$ be any point in $P \setminus P(v')$,
    then $\norm{p - q} > (\Delta/dn) (1+\lambda)^{t+\Gamma}$.
\end{enumerate}
\end{lemma}

\begin{proof}
Note that for any two nodes of the aggregation tree $T_P$,
either the associated point sets of the two nodes are disjoint,
or one of the node is a descendant of the other.
Since in Algorithm \ref{alg-merge},
we avoid putting a node in $B^+_t$ when the node's descendant is already in $B^+_t$.
Thus the associated sets of the nodes in $B^+_t$ are disjoint.
Clearly there exists no more than one $v \in B^+_t$ such that $p \in P(v)$. 

In the following, we let $v_0$ be the leaf of $T_P$ such that $P(v_0) = \{p\}$.
Let $v'$ be the farthest ancestor of $v_0$ in $T_P$ such that $s(v')/\lambda \leq (1+\lambda)^t$.
We show that either $v'\in B^+_t$,
or $\norm{p - q} > (\Delta/dn) (1+\lambda)^{t+\Gamma}$ for any point $q$ in $P \setminus P(v')$.  
We assume that $v'$ is not the root of $T_P$, since otherwise we are done.
Let $v''$ be the parent of $v'$ in $T_P$.
Let $\Delta' := \Delta \cdot (1+\lambda)^{\Gamma}$ (\emph{i.e.}, $\Delta'$ is the third parameter passed to the call to
RangeCover in Step 1 of Algorithm \ref{alg-multi}).
Clearly, $s(v')/\lambda \leq (1+\lambda)^t < s(v'')/\lambda$.
We consider three cases:
(1) $s(v'')/\lambda > (1+\lambda)^t \geq s(v')/\lambda \geq s(v'')/\Delta'$;
(2) $s(v'')/\lambda > (1+\lambda)^t \geq s(v'')/\Delta' > s(v')/\lambda$;
and (3) $s(v'')/\lambda > s(v'')/\Delta' > (1+\lambda)^t \geq s(v')/\lambda$.

\noindent\textbf{Case (1) and (2).} Clearly, in this case,  $v' \in B_t$.
Note that since $s(v')\lambda \leq (1+\lambda)^t$,
we know that any proper descendant of $v'$ will not be put in a bucket $B_{t'}$ for any $t' \geq t$.
Thus, we have $v' \in B^+_t$.

\noindent\textbf{Case (3).} Assume that $v' \not\in B^+_t$.
From the property of the aggregation tree,
we know that for any $q \in P \setminus P(v')$,
$\norm{q - p} \geq s(v'')/nd > (\Delta'/nd)(1+\lambda)^t = (\Delta/nd)(1+\lambda)^{t+\Gamma}$.
This completes the proof.
\end{proof}

Although the sketch does not fully cover $P$, in many problems (including  AIFP) 
these points are either negligible or easy to handle
by other means due to their special properties.

From \cite{huang2019-tdrc}, we know that the running time of Algorithm \ref{alg-range} 
and the space complexity of the output data structure is
$O_{\lambda}(n \log n \Delta)$ (where the hidden constant in the big-O notation
depends only on $\lambda$).
Algorithm \ref{alg-merge} essentially merges  $\Gamma+1$ consecutive
buckets $B_t$, $B_{t+1}$, $\ldots$, $B_{t+\Gamma}$
created by Algorithm \ref{alg-range} into one bucket $B^+_t$.
Thus, we have the following lemma.

\begin{lemma}
\label{lm-rangesize}
Excluding the time it takes for $\mathcal{A}$ to process each $B^{+}_t$ in Step 4,
the running time of Algorithm \ref{alg-multi}
and the total number of nodes in all buckets is
$O_{\lambda}(\Gamma^2 n \log n\Delta)$,
where the hidden constant in big-O notation
depends only on $\lambda$.
\end{lemma}

We conclude this subsection by providing a key lemma showing that,
given a constrained AIPF query $(B,q)$ satisfying constraint $(r_{B}, d_{min}, d_{max})$
with $d_{min} \leq r_{B} \leq d_{max}$,
if there is a bucket $B^{+}_{t}$ such that $(1+\lambda)^{t+1}/(1-\lambda) < d_{min} < d_{max} \leq (1+\lambda)^{t+\Gamma} - (2+2/\lambda)r_B$
 (\emph{i.e.}  the range $[d_{min}, d_{max}]$ falls in interval $((1+\lambda)^{t}, (1+\lambda)^{t+\Gamma+1}]$
 with some gap),
then, with an easy-to-handle exception,
an AIFP to $q$ in $B^{+}_{t}$ (by viewing every node of $B^{+}_{t}$ as one point)
in (slightly enlarged) range $B$ is also
an AIFP of $q$ to $P \cap B$.
Formally, let $re(B^{+}_{t}) := \{ re(v) \mid v \in B^{+}_{t}\}$.
Let $p_t$ be the farthest point to $q$ in $re(B^{+}_{t}) \cap B(1+\lambda)$
if $re(B^{+}_{t}) \cap B(1+\lambda) \not= \emptyset$,
and $p$ be a $(\lambda/6,\lambda/6)$-AIFP of $q$ in $re(B^{+}_{t}) \cap B(1+\lambda)$
\footnote{Note $p$ could be NULL here. This could happen when bucket $B^{+}_t$ is empty or $re(B^{+}_{t}) \cap B(1+\lambda) = \emptyset$. }.
Let $p_N$ be a $(1+\lambda)$-approximate nearest neighbor of $q$ in $P$.

\begin{lemma}
\label{lm-multi}
One of the following holds: 
(1) $p_N$ is a $(2\lambda,2\lambda)$-AIFP of q in $P \cap B$, or
(2) $p_t$ exists and $(1+\lambda)^t \leq \norm{q - p_t} \leq (1+\lambda)^{t+\Gamma+1}$, 
    and $p$ is a $(2\lambda,2\lambda)$-AIFP of $q$ in $P \cap B$.
\end{lemma}

\begin{proof}
We first consider the case where there exists a point $p \in P \cap B$,
such that there does NOT exist $v(p) \in B^+_t$ with $p \in P(v(p))$.
We will show that in this case, item (1) of the lemma holds.

If such a $p$ exists,
by Lemma \ref{lm-range2},
we know that there exists an aggregation node $v'$ such that $p \in P(v')$,
$s(v') \leq \lambda (1+\lambda)^t < \lambda d_{min}$,
and for any $p' \in P \setminus P(v')$,
$\norm{p - p'} > (\Delta/dn) (1+\lambda)^{t+\Gamma} \geq 4 (d_{max} + (2+2\lambda)r_B)
\geq 4 (\norm{p - q} + (2+2\lambda)r_B)$.
Since $p$ lies in $B$, it is not hard to see that $p'$ does not lie in $B(1+\lambda)$.
This implies if a point lies in $P \cap B$,
then it is in $P(v')$.
Let $p_f$ be the farthest point to $q$ in $B \cap P$.
Thus $p_f \in P(v')$.
For any point $p'' \in P(v')$, we have
$\norm{p'' - p_f} \leq s(v') \leq  \lambda d_{min} \leq \lambda \norm{q - p_f}$,
therefore $\norm{q - p''} \geq \norm{q - p_f} - \norm{p'' - p_f} \geq (1-\lambda) \norm{q - p_f}$.
Also we have 
$\norm{p'' - p_f} \leq  \lambda d_{min} \leq \lambda r_B$.
Since $p_f \in B$, it is not hard to see
$p'' \in B(1+\lambda)$.
We have shown that any point $p'' \in P(v')$ is an $(\lambda,\lambda)$-AIFP of $q$ in $B \cap P$.
Next, we show that $P_N \in p(v)$.
Recall that for any $p' \in P \setminus P(v')$,
$\norm{p - p'} \geq 4 (\norm{p - q} + (2+2\lambda)r_B)$.
We have $\norm{q - p'} \geq \norm{p - p'} - \norm{q - p} \geq 3(\norm{p - q} + (2+2\lambda)r_B)$.
On the other hand, 
since $p_f \in B$, we have $\norm{p_f - q} \leq \norm{p_f - p} + \norm{p - q} \leq 2r_B + \norm{p - q}$.
Therefore, $\norm{p_f - q} < \norm{p - p'}/2$, thus
$p'$ cannot be a $(1+\lambda)$-approximate nearest neighbor of $q$ in $P$.
We have proved item (1) of the lemma in this case.

In the following,
we assume that for any $p \in P \cap B$,
there exist $v(p) \in B^+_t$ with $p \in P(v(p))$.
We prove item (2) of the lemma holds in this case.

Let $p_f$ be the farthest point to $q$ in $B \cap P$.
(Note we can assume that $p_f$ exists, otherwise $B \cap P = \emptyset$, which means any point in $P$
is an AIFP of $q$, and item (1) of the lemma trivially holds.)
We first show $p_t$ exists and $\norm{q - p_t} \geq (1+\lambda)^{t}$.
Let $v_f \in B^{+}_t$ be the node such that $p_f \in P(v_f)$.
We have $\norm{p_f - re(v_f)} \leq s(v_f) \leq \lambda (1+\lambda)^t \leq \lambda r_B$.
Since $p_f \in B$, we have $re(v_f) \in B(1+\lambda)$.
This means $B(1+\lambda) \cap re(B^{+}_t)$ is non-empty, thus $p_t$ exists.
Also, since $re(v_f) \in re(B^{+}_t)$,
from the definition of $p_t$,
we have $\norm{q - p_t} \geq \norm{q - re(v_f)}$.
Thus $\norm{q - p_t} \geq \norm{q - re(v_f)} \geq \norm{q - p_f} - \norm{p_f - re(v_f)}
\geq d_{min} -  \lambda (1+\lambda)^t \geq d_{min} -  \lambda d_{min} \geq (1-\lambda) (1+\lambda)^{t+1} / (1-\lambda) \geq 
(1+\lambda)^{t}$.

Next we show $\norm{q - p_t} \leq (1+\lambda)^{t+\Gamma+1}$.
Clearly both $p_f$ and $p_t$ are in $B(1+\lambda)$,
thus $\norm{p_t - p_f} \leq (2 + 2\lambda)r_B$.
We have $\norm{q - p_t} \leq \norm{q - p_f} + \norm{p_t - p_f} \leq d_{max} + (2 + 2\lambda)r_B \leq (1+\lambda)^{t+\Gamma+1}$.

Finally, we prove that the point $p$, which is a $(\lambda/6, \lambda/6)$-AIFP of $q$ in $P \cap B(1+\lambda)$,
is a $(2\lambda, 2\lambda)$-AIFP of $q$ in $P \cap B$.
Note $p \in B((1+\lambda) \cdot (1+\lambda/6)) \subset B(1+2\lambda)$.
Also, 
note $\norm{p_f - re(v_f)} \leq s(v_f) \leq \lambda (1+\lambda)^t \leq \lambda d_{min} \leq \lambda \norm{q - p_f}$,
Thus $\norm{q - re(v_f)} \geq  \norm{q - p_f} - \norm{p_f - re(v_f)} \geq (1-\lambda) \norm{q - p_f}$.
For any $p_B \in p \cap B$,
We have $\norm{q - p} \geq (1-\lambda/6) \norm{q - re(v_f)} \geq (1-\lambda/6)(1 - \lambda) \norm{q - p_f}
\geq (1-\lambda/6)(1 - \lambda) \norm{q - p_B} \geq (1-2\lambda) \norm{q - p_B}$.
We have proved $p$ is a $(2\lambda,2\lambda)$-AIFP of $q$ in $P \cap B$.

\end{proof}

The above lemma implies that,
in Algorithm \ref{alg-multi},
if routine $\mathcal{A}$ builds a constrained AIFP data structure for
farthest distance lies in interval $[(1+\lambda)^t, (1+\lambda)^{t+\Gamma+1}]$,
then either this data structure can be used to answer any query with constraint $(r_{B}, d_{min}, d_{max})$
(with other parameters, like $r_{B}$ and the approximate factors for constrained AIFP, set properly),
or the AIFP query can be solved easily using a nearest neighbor search.
In the following section, we will show how to build
a general AIFP query data structure through multi-scale construction
by selecting appropriate parameters.
With the multi-scale data structure (together with some auxiliary data structures), 
we can answer an AIFP query by (1) obtaining a rough estimation
of the farthest distance, and (2) querying the bucket corresponding to the estimated range.

\section{General AIFP Query}
\label{sec-aifp}

In this section, we present a general $(\epsilon, \gamma)$-AIFP query scheme.
We first establish several facts for better understanding of the query scheme.
In the following, let $B$ be a closed ball 
with radius
$r_{B} > 0$
and $q$ be an arbitrary point. Let $0 < \epsilon <1$ and $0 < \gamma < 1$ be any pair of constants and 
$\lambda := \min(\epsilon,\gamma)/512$.
We assume that $r_{B}$ is \emph{$\lambda$-aligned}, which means that $r_{B} = (1+\lambda)^{t}$ for some integer $t$.
The alignment assumption makes it easier to implement the multi-scale construction.
Note that if $r_{B}$ is not $\lambda$-aligned, we can always enlarge $B$ a little to make $r_{B}$ aligned and 
still obtain a good approximation with carefully chosen parameters.

Our main idea 
is to convert each query $(B,q)$ into one or more AIFP queries $(B',q)$ 
such that it is possible to find a lower bound $d'_{min}$ and an upper bound $d'_{max}$ on the farthest distance between $q$ and a point in $B' \cap P$.
With such bounds, the AIFP can then be found by querying a  pre-built constrained AIFP data structure. 
To ensure efficiency, the gap between $d_{min}$ and $d_{max}$ cannot be too large, {\em i.e.,} 
$d_{max}/d_{min}$ should be bounded by a polynomial of $n$ and $d$.
Since the complexity of a constrained data structure depends on $d'_{max}/d'_{min}$,
a small gap will also enable us to control the size of the data structure.
We start with a simple claim. 

\begin{claim}
\label{cl-far}
If the distance between $q$ and the center $o_{B}$ of $B$ is
very large compared to $r_{B}$, \emph{i.e.} $\norm{q - o_{B}} \geq (3+\gamma) \epsilon^{-1} r_{B}$,
then any point in $B(1+\gamma/2)$ is an $(\epsilon,\gamma)$-AIFP of $q$ in $P \cap B$.
\end{claim}

This claim suggests that we can safely assume that the farthest distance between $q$ and $P \cap B$ is not too large (compared to $r_{B}$), as
otherwise the AIFP can be easily found. This helps us establish an upper bound on the farthest distance.
In the following, we let $d_{max} :=  (4+2\gamma) \epsilon^{-1} r_{B}$,
and assume that $\norm{q - o_{B}} \leq (3+\gamma) \epsilon^{-1} r_{B}$. 
From simple calculation, this implies that 
for any $p \in P \cap B(1+\gamma)$, we have $\norm{p-q} \leq d_{max}$.

Next, we try to find a lower bound for the farthest distance.
This process is much more complicated.
We need an aggregation tree $T_P$ with distortion polynomial
$\mathcal{P}_{HST}(n,d) = nd$. Later, we will use this $T_P$ to  construct the query data structure.
Let $p_N \in P$ be a $2$-nearest neighbor of $q$ in $P$  (\emph{i.e.} for any $p' \in P$,
$\norm{p_N - q} \leq 2 \norm{p'-q}$), and 
$r_N$ be its distance to $q$ ({\em i.e.,} $r_N := \norm{p_N - q}$). Denote by $v$ the lowest (closest to a leaf)
node of $T_P$ such that $p_N \in P(v)$ and $s(v) + r_{N} \leq \gamma r_{B}/64$.
Note that such a node may not exist (\emph{i.e.} $r_{N} > \gamma r_{B}/64$), and we will show that the AIFP can be found easily in such a case. We have the following claims.

\begin{claim}
\label{cl-out}
If node $v$ does not exist, for any $p \in P$,
$\norm{q - p} \geq  \gamma r_{B}/256 nd$;
Otherwise, for any $p_{out} \in P \setminus P(v)$,
$\norm{q - p_{out}} \geq  \gamma r_{B}/256 nd$.
\end{claim}

The following two claims assume that $v$ does exist.

\begin{claim}
\label{cl-case1}
If $\norm{o_{B} - q} > (1+\gamma/64)r_{B}$, then for every $p \in P(v)$,
$p \not\in B$.
\end{claim}

For any positive real number $x$, let $[x]_{\lambda}$ denote the smallest real number
that can be written as
$(1+\lambda)^t$ for some integer $t$ 
such that $(1+\lambda)^t \geq x$.

\begin{claim}
\label{cl-case2}
Let $r^{-} := r_{N} + s(v)$ and $r^{(2)}_B := [(1+\gamma/16)r_B]_{\lambda}$. 
If $\norm{o_{B} - q} \leq (1+\gamma/64)r_{B}$, and $P \cap B \not= \emptyset$
then  
\begin{enumerate}
    \item If $r^{-} > \gamma r_{B}/512 n^2 d^2$, 
    there exists $p_1 \in \mathcal{B}(o_B, r^{(2)}_B) \cap P$ 
    such that $\norm{p_1 - q} \geq \gamma r_{B}/2048 n^3 d^3$.
    \item If $r^{-} \leq \gamma r_{B}/512 n^2 d^2$, the following holds, where  $r^{(3)}_B := [r^{-}]_{\lambda}$.    (a) $\mathcal{B}(q, r^{(3)}_B) \subset B(1+\gamma/32)$;
        (b) For every $p \in P(v)$, $p \in \mathcal{B}(q, r^{(3)}_B)$ and for every $p \in P \setminus P(v)$, $p \not\in \mathcal{B}(q, (1+\gamma)r^{(3)}_B)$; 
        (c) There exists $p' \in P(v)$ such that $\norm{p' - q} \geq r^{(3)}_B/8nd$. 
\end{enumerate}
\end{claim}

The above facts allow us to  reduce the AIFP query
to a constrained query with an estimated lower bound on the farthest distance.
We assume $P \cap B \not= \emptyset$, otherwise the query becomes trivial.
Let $p_{F}$ be the farthest point to $q$ in $P \cap B$. Note that there are three possible cases: $p_{F} \in P(v)$, $p_{F} \in P \setminus P(v)$, or $v$ does not exist.
If $p_{F} \in P \setminus P(v)$ or $v$ does not exist,
from Claim \ref{cl-out} we know that $(B,q)$ satisfies the constraint of $(r^{(1)}_{B}, d^{(1)}_{min}, d^{(1)}_{max})$, where $r^{(1)}_B := r_B$, $d^{(1)}_{min} := \gamma r_{B}/256 nd$ and $d^{(1)}_{max} := d_{max}$. 
Thus, these two cases can be captured by performing a constrained AIFP query to the pre-built data structures.
Also note that $d^{(1)}_{max}/d^{(1)}_{min} = 256(4+2\gamma) nd/\epsilon$, which is a polynomial of $n$ and $d$.
This means that the space complexity of the pre-built data structures can be bounded.

Next,  we consider the case of $p_{F} \in P(v)$. 
Note that if $\norm{o_{B} - q} > (1+\gamma/64)r_{B}$,
from Claim \ref{cl-case1} we know that all points in $P(v)$ are not in $B$,  which is a contradiction.
Thus, we can assume that $\norm{o_{B} - q} \leq (1+\gamma/64)r_{B}$, which leads to two sub cases
discussed in Claim \ref{cl-case2}:
(1) $r^{-} > \gamma r_{B}/512 n^2 d^2$, and (2) $r^{-} \leq \gamma r_{B}/512 n^2 d^2$.

For sub-case (1), we first let $d^{(2)}_{min} := \gamma r_{B}/2048 n^3 d^3$. Then, from Claim \ref{cl-case2} we know that 
the farthest distance between $q$ and a point in $P \cap \mathcal{B}(o_B, r^{(2)}_B)$ is at least $d^{(2)}_{min}$.
From previous discussion and the fact that $\mathcal{B}(o_B, r^{(2)}_B) \subset B(1+\gamma)$,
we also know that this farthest distance is at most $d_{max}$.
Let $d^{(2)}_{max} := d_{max}$.
The problem of finding the farthest point $p'_F$ to $q$ from all points in $P \cap \mathcal{B}(o_B, r^{(2)}_B)$
satisfies the constraint of $(r^{(2)}_B, d^{(2)}_{min}, d^{(2)}_{max})$, and is thus solvable by querying a constrained
AIFP data structure. Since the ball range $\mathcal{B}(o_B, r^{(2)}_B)$ only slightly enlarges $B$ and
$\mathcal{B}(q, r^{(2)}_B) \subset B(1+\gamma)$,
with properly selected approximation factors, it can be shown that the found $p'_F$ is a $(\epsilon,\gamma)$-AIFP of $q$ in $B \cap P$.
In this case we also have $d^{(2)}_{max}/d^{(2)}_{min}$ bounded by a polynomial $2048(4+2\gamma) n^3 d^3/\epsilon$.

For sub-case (2), we let $r^{(3)}_B := d^{(3)}_{max} := r^{(3)}_B$ and $d''_{min} := r^{(3)}/8nd$.
From Claim \ref{cl-case2}, we know that if $p_{F} \in P(v)$,  $p_{F}$ is also the farthest point to $q$ in 
$P \cap B \cap \mathcal{B}(q, r^{(3)}_B)$.
An $(\epsilon,\gamma)$-AIFP of $q$ in $P \cap B$ can then be found by identifying an AIFP of $q$ in 
$P \cap \mathcal{B}(q, r^{(3)}_B)$. Note that this query satisfies the constraint of $(r^{(3)}_B, d^{(3)}_{min}, d^{(3)}_{max})$,
with $d^{(3)}_{max}/d^{(3)}_{min}$ bounded by a polynomial of $n$ and $d$.

From the above discussions, we know that 
given any AIFP query $(B,q)$ (with radius of $B$ aligned),
it is possible to reduce it to a constrained AIFP query with  constraint  $(r^{(i)}_B, d^{(i)}_{min}, d^{(i)}_{max})$, which  
satisfies the inequality of 
$d^{(i)}_{min} \leq r^{(i)}_B \leq d^{(i)}_{max}$, 
 $d^{(i)}_{max}/d^{(i)}_{min}$ is bounded by a polynomial $\mathcal{P}_{gap}(n,d) := 2048(4+2\gamma) n^3 d^3/\epsilon$,
and $r^{(i)}_B$ is aligned. 
In the following, we will show how to use multi-scale construction to build a data structure that supports all such 
constrained AIFP queries.

\subsection{Multi-scale Construction for General AIFP Query}

In this subsection, we show how to build a Multi-Scale data structure 
to  answer general AIFP queries.
Our goal is to choose the appropriate parameters $\Gamma \geq 1$, $0<\lambda<1$ and $\Delta \geq 4nd$ 
for Algorithm \ref{alg-merge} so that for
every AIFP query with constraint $(r^{(i)}_B, d^{(i)}_{min}, d^{(i)}_{max})$, 
there exists a bucket $B^{+}_t$ whose range $( (1+\lambda)^t, (1+\lambda)^{t+\Gamma+1} ]$
wholly covers the interval $ [d^{(i)}_{min}, d^{(i)}_{max}] $,
and the constrained AIFP data structure built for the bucket 
with $d_{min} := (1+\lambda)^t$ and $d_{max} := (1+\lambda)^{t+\Gamma+1}$ 
can be used to answer the query.
Note that we have already defined $\lambda := \min(\epsilon,\gamma)/512$ and  $\Delta := 4nd$. 
The remaining task is to determine the value of $\Gamma$.

Observe that $d^{(i)}_{max}/d^{(i)}_{min}$ is bounded by a polynomial $\mathcal{P}_{gap}(n,d) := 2048(4+2\gamma) n^3 d^3/\epsilon$.
Let $\Gamma' := \lceil \log_{1+\lambda} \mathcal{P}_{gap}(n,d) \rceil$, and 
 $\Gamma_{L} \geq \Gamma'$ and $\Gamma_{R} \geq \Gamma'$ be integer parameters to be determined later.
Denote by $\Gamma$ the sum of $\Gamma_{L}$ and $ \Gamma_{R}$, {\em i.e.,}  $\Gamma := \Gamma_{L} + \Gamma_{R}$.
For every integer $t$, define $r_{mid}(t) := (1+\lambda)^{t+\Gamma_{L}}$. 
Therefore, we have $ r_{mid}(t)/(1+\lambda)^t \geq (1+\lambda)^{\Gamma_{L}}  \geq \mathcal{P}_{gap}(n,d)$
and $ (1+\lambda)^{t+\Gamma+1}/r_{mid}(t) \geq (1+\lambda)^{\Gamma_{R}}  \geq \mathcal{P}_{gap}(n,d)$.
For any AIFP query $(B,q)$ with constraint $(r^{(i)}_B, d^{(i)}_{min}, d^{(i)}_{max})$,
it is always possible to find a bucket $B^{+}_t$ such that $r^{(i)}_B = r_{mid}(t)$. 
Clearly, interval $( (1+\lambda)^t, (1+\lambda)^{t+\Gamma+1} ]$
wholly covers the interval $ [d^{(i)}_{min}, d^{(i)}_{max}] $.
If a constrained AIFP data structure is constructed for $B^{+}_t$ with constraint
$((1+\lambda)r_{mid}(t), (1+\lambda)^{t}, (1+\lambda)^{t+\Gamma+1})$,
it can be used to answer the AIFP query $(B,q)$. (See Lemma \ref{lm-multi}.)


To summarize the above discussions,
we set the parameters of Algorithm \ref{alg-multi} as the following.
The algorithm then produces the data structure for $(\epsilon, \gamma)$-AIFP query.
Assume that a real number $0<\delta<1$ is given and we would like to achieve $1-\delta$ query success probability.

\begin{itemize}
    \item $\lambda := \min(\epsilon, \gamma)/512, 
    \Gamma_{L} := \Gamma' + \lceil \log_{1+\lambda} 8 \rceil, 
    \Gamma_{R} := \Gamma' + \lceil \log_{1+\lambda} 8 \rceil,
    \Gamma := \Gamma_{L} + \Gamma_{R}, \Delta := 4nd$.
    \item Routine $\mathcal{A}$: Given a non-empty bucket $B^+_t$, $\mathcal{A}$,
    it uses Algorithm \ref{alg-build-caifp} to
    creates a constrained $(\lambda/6, \lambda/6)$-AIFP data structure for point set $re(B^+_t)$
    for constraint $( (1+\lambda) r_{mid}(t), (1+\lambda)^t, (1+\lambda)^{t+\Gamma+1} )$,
    with success probability at least $1 - \delta/4$.
\end{itemize}

Note that we let $\Gamma_{L} := \Gamma_{R} := \Gamma' + \lceil \log_{1+\lambda} 8 \rceil$. 
This allows more gap when fitting the interval $ [d^{(i)}_{min}, d^{(i)}_{max}] $
in  $( (1+\lambda)^t, (1+\lambda)^{t+\Gamma+1} ]$, which is required by Lemma \ref{lm-multi}.







With the multi-scale data structure constructed by Algorithm \ref{alg-multi} using the above parameters, 
we are able to answer any general AIFP query
by reducing it to constrained AIFP queries with constraints  $(r^{(i)}_B, d^{(i)}_{min}, d^{(i)}_{max})$, $i=1,2,3$.
The detailed algorithm is described as the following Algorithm \ref{alg-queryAIFP} 
(for align query ball radius) and
Algorithm \ref{alg-queryAIFP-full} (for any radius).
Besides the multi-scale structure built using the
parameters $\lambda,\Gamma,\Delta$ and $\mathcal{A}$ as described above,
a $(1+\lambda)$-nearest neighbor data structure from \cite{har2012approximate}, with success probability $1 - \delta/4$ is also used.

\begin{algorithm}[th]
\caption{Query-AIFP-ALIGNED$(B,q)$}

\textbf{Input:} A query ball $B$ with $\lambda$-aligned radius $r_B$ and centered at $o_{B}$. A query point $q$. \\
\textbf{Output:} An $(\epsilon,\gamma/2)$-AIFP of $q$ in $P \cap B$.
\label{alg-queryAIFP}

\begin{algorithmic}[1]

\State{ If $\norm{q - o_{B}} \geq (3+\gamma) \epsilon^{-1} r_{B}$.
Find a $(1+\gamma/2)$-approximate nearest neighbor $p_o \in P$ of $o_{B}$.
(Note this can be done using the $(1+\lambda)$-approximate nearest neighbor data structure,
since $1+\lambda < 1 + \gamma/2$.)
If $p_o \in B(1+\gamma/2)$, return $p_o$, otherwise return NULL.
}

\State{
Initialize a point set $\mathcal{Q} := \{ \}$.
}

\State{
Let $t_1$ be the integer such that $r_{mid}(t_1) = r_B$.
Query the constrained AIFP data structure for bucket $B^{+}_{t_1}$ with $(B,q)$ (using {\bf Algorithm} \ref{alg-query-caifp}).
Let $p_1$ be the result and put $p_1$ into $\mathcal{Q}$.
}

\State{
Find $p_N \in P$ which is a $(1+\lambda)$-nearest neighbor of $q$ in $P$.
Denote $r_N := \norm{p_N - q}$, and put $p_N$ into $\mathcal{Q}$.
}

\State{
Let $v$ be the lowest (closest to a leaf)
node of $T_P$ such that $p \in P(v)$ and $s(v) + r_{N} \leq r_{B}/64\gamma$.
(Such $v$ can be found by binary search on the tree path.)
If $v$ does not exist, \textbf{GOTO} step 7.
}

\State{
Denote $r^{(2)}_B := [(1+\gamma/16)r_B]_{\lambda}$.
Let $t_2$ be the integer such that $r_{mid}(t_2) = r^{(2)}_B$.
Query the constrained AIFP data structure for bucket $B^{+}_{t_2}$
with $(\mathcal{B}(o_B, r^{(2)}_B), q)$.
Let $p_2$ be the result and put $p_2$ into $\mathcal{Q}$.
}

\State{
Denote $r^{(3)}_B := [r_{N} + s(v)]_{\lambda}$.
Let $t_3$ be the integer such that $r_{mid}(t_3) = r^{(3)}_B$.
Query the constrained AIFP data structure for bucket $B^{+}_{t_3}$
with $(\mathcal{B}(q, r^{(3)}_B), q)$.
Let $p_3$ be the result and put $p_3$ into $\mathcal{Q}$.
}

\State{
From $\mathcal{Q}$, pick the one that lies in $B(1+\gamma/2)$ and the distance to $q$ is the farthest.
If no such point exists, return NULL.
}

\end{algorithmic}
\end{algorithm}

\FloatBarrier

We then present AIFP query algorithm for $B$ with arbitrary radius.
The algorithm simply enlarges $B$ a little to make the radius aligned,
then call the above Algorithm \ref{alg-queryAIFP} to handle it.

\begin{algorithm}[th]
\caption{Query-AIFP$(B,q)$}

\textbf{Input:} A query ball $B$ with radius $r_B$ and center $o_B$. A query point $q$. \\
\textbf{Output:} An $(\epsilon,\gamma)$-AIFP of $q$ in $P \cap B$.
\label{alg-queryAIFP-full}

\begin{algorithmic}[1]

\State{
Let $r'_B := [r_B]_{\lambda}$.
Make a query $(\mathcal{B}(o_B, r'_B),q)$ using Algorithm \ref{alg-queryAIFP},
and return the result.
}

\end{algorithmic}
\end{algorithm}

It is worth noting that since we reduce a general query
to at most three constrained AIFP queries,
from Lemma \ref{lm-aifp} and the fact that the ratio $d_{max}/d_{min}$ 
satisfies
$d_{max}/d_{min} = (1+\lambda)^{t+\Gamma+1}/(1+\lambda)^{t}
= (1+\lambda)^{\Gamma+1}$, which is bounded by a polynomial of $n,d$,
we know that $\log d_{max}/d_{min}$ is $O_{\epsilon,\gamma}(\log nd)$ and
the query time is sub-linear.
In the next subsection we provide the detailed analysis of the algorithms and prove Theorem \ref{thm-main-aipf}.

\subsection{The Analysis of the AIFP Algorithms}

In this subsection, we analyze the algorithms for AIFP and prove Theorem \ref{thm-main-aipf},
the main result of this paper.
We first prove the correctness of Algorithm \ref{alg-queryAIFP}.

\begin{lemma}
With probability at least $1 - \delta$, Algorithm \ref{alg-queryAIFP}
returns a $(\epsilon,\gamma/2)$-AIFP of $q$ in $P \cap B$.
\end{lemma}

\begin{proof}
We assume that all the constrained AIFP queries and nearest neighbor queries are successful.
It is easy to see that this happens with probability at least $1 - \delta$.

We decompose the discussion into several cases.

\textbf{Case 1:} $\norm{q - o_B} \geq (3+\gamma) \epsilon^{-1} r_B$.
In this case, the algorithm return a point $p_o \in P \cap B(1 + \gamma/2)$ in Step 1, 
or return NULL, which only happens
when there does not exist points in $P \cap B$ and NULL is a correct answer in this case.
From our earlier proof of Claim \ref{cl-far}, we have shown any point in $P \cap B(1 + \gamma/2)$
is a $(\epsilon, \gamma/2)$-AIFP of $q$ in $P \cap B$.
This case is proved.

In the following, we assume that $\norm{q - o_B} < (3+\gamma) \epsilon^{-1} r_B$.
We would show that in each of the following cases,
a correct $(\epsilon,\gamma/2)$-AIFP of $q$ in $P \cap B$ would be put into the set $\mathcal{Q}$.
Thus an $(\epsilon,\gamma/2)$-AIFP of $q$ in $P \cap B$ would be returned in Step 8.

\textbf{Case 2:} The node $v$ in Step 5 does not exists, or $p_F \in P \setminus P(v)$ where $p_F$ is
the farthest point to $q$ in $P \cap B$.
From the discussion in Section \ref{sec-aifp}, query $(B,q)$ satisfies a constraint
$(r^{(1)}_B,d^{(1)}_{min},d^{(1)}_{max})$ so that it can be answered by the bucket $B^+_{t_1}$ in Step 3.
From Lemma \ref{lm-multi},
either the point $p_1$ obtained in Step 1, or $p_N$ obtained in Step 4, would be a $(2\lambda, 2\lambda)$-AIFP of $q$ in $P \cap B$,
and is clearly a $(\epsilon,\gamma/2)$-AIFP of $q$ in $P \cap B$.

In the following, we exclude the scenario where $p_N$ is a $(2\lambda, 2\lambda)$-AIFP of $q$ for a constrained query we
make in Step 6 or 7.
Indeed this scenario is handled well since $p_N$ is already put into $\mathcal{Q}$.

\textbf{Case 3:} $v$ exists and $p_F \in P(v)$.
From the discussion in Section \ref{sec-aifp},
one of the following holds.
\begin{enumerate}
    \item[(a)] Item 1 of Claim \ref{cl-case2} holds, 
    and the query $(\mathcal{B}(o_B, r^{(2)}_B), q)$
    satisfies a constraint
    $(r^{(2)}_B,d^{(2)}_{min},d^{(2)}_{max})$ so that it can be answered by the bucket $B^+_{t_2}$ in Step 6;
    \item[(b)] Item 2 of Claim \ref{cl-case2} holds, and the query $(\mathcal{B}(q, r^{(3)}_B), q)$
    satisfies a constraint
    $(r^{(3)}_B,d^{(3)}_{min},d^{(3)}_{max})$ so that it can be answered by the bucket $B^+_{t_3}$ in Step 7.
\end{enumerate}

For sub-case (a),
$p_2$ is a $(2\lambda, 2\lambda)$-AIFP of $q$ in $P \cap \mathcal{B}(o_B, r^{(2)}_B)$.
Thus $p_2 \in \mathcal{B}(o_B, (1+2\lambda) r^{(2)}_B)$.
From some simple calculation, we know that $(1+2\lambda) r^{(2)}_B \leq (1+\gamma/2) r_B$.
Also $B \subset \mathcal{B}(o_B, r^{(2)}_B)$.
Thus $p_2$ is a $(\epsilon,\gamma/2)$-AIFP of $q$ in $P \cap B$.

Finally, for sub-case (b), $p_3$ is a $(2\lambda, 2\lambda)$-AIFP of $q$ in $P \cap \mathcal{B}(q, r^{(3)}_B)$.
Note since $p_F \in P(v) \subset \mathcal{B}(q, r^{(3)}_B)$, we have 
$\norm{q - p_3} \geq (1-2\lambda)\norm{q - p_F} \geq (1-\epsilon) \norm{q - p_F}$.
Also $\norm{q - p_3} \leq (1+2\lambda)r^{(3)}_B \leq (1+\gamma)r^{(3)}_B$,
thus $p_3 \in P(v) \subset \mathcal{B}(q, r^{(3)}_B) \subset B(1+\gamma/2)$.
Thus $p_3$ is a $(\epsilon,\gamma/2)$-AIFP of $q$ in $P \cap B$.

We have proved that in any case, some point in $\mathcal{Q}$ would be a $(\epsilon, \gamma/2)$
of $q$ in $P \cap B$. The proof is complete.

\end{proof}

In the following we prove the correctness of Algorithm \ref{alg-queryAIFP-full}.
We show that Algorithm \ref{alg-queryAIFP-full} returns $(\epsilon,\gamma)$-AIFP
of $q$ in $P \cap B$ with probability at least $1-\delta$.

\begin{proof}

Assume the call to Algorithm \ref{alg-queryAIFP} is successful, which happens with probability at
least $1 - \delta$.
Let $p$ be the result of the call to Algorithm \ref{alg-queryAIFP}.
Then $\norm{o_B - p} \leq (1+\gamma/2) r'_B \leq (1+\gamma/2)(1+\lambda) r_B \leq (1+\gamma) r_B$.
Also for any $p' \in B \cap P \subset  \mathcal{B}(o_B, r'_B) \cap P$,
we have $(1-\epsilon)\norm{q - p'} \leq \norm{q - p}$.
Thus $p$ is an $(\epsilon,\gamma)$-AIFP
of $q$ in $P \cap B$.

\end{proof}

In the following, we analyze the space and time complexity of the query scheme.

\textbf{Pre-processing Time and Space.}
The AIFP query data structure consists of three parts: An aggregation tree,  a multi-scale data structure
and a $(1+\lambda)$-nearest neighbor data structure.
Using the technique in \cite{har2012approximate},
an aggregation tree takes $O(dn\log^2 n)$ time to build and takes $O(dn)$ space.
A $(1+\lambda)$-nearest neighbor data structure takes $O_{c}(dn^{1+1/c} \cdot \text{Polylog}(n))$ time to build and
uses $O_{c}(dn^{1+1/c} \cdot \text{Polylog}(n))$ space, where $c := 1 + \lambda$ and the constant hidden by
the $O_{c}(\cdot)$ notation depends only on $c$.
For the multi-scale data structure,
note that the total number of nodes in all buckets is $O_{\lambda}(\Gamma^2 n \cdot \log (n\Delta))$.
From the values of $\lambda, \Gamma, \Delta$,
we know the total number of nodes is $O_{\epsilon,\gamma}(n \cdot \log n \cdot \log^2 (nd))$.
For each bucket, we build a constraint $(\lambda/6,\lambda/6)$-AIFP data structure with query success probability
$1 - \delta/4$,
and we have $\log (d_{max}/d_{min}) = O_{\epsilon,\gamma}(\log nd)$.
From Lemma \ref{lm-aifp},
the total preprocessing time and space for these data structures would be 
$O_{\epsilon,\gamma}(d \cdot [n \cdot \log n \cdot \log (nd)]^{1 + \rho} \cdot \log \delta^{-1} \cdot \log^2 nd)$,
for some constant $0<\rho<1$ which depends on $\epsilon, \gamma$.
The space/time bound can be simplified as $O_{\epsilon,\gamma}(d n^{1+\rho} \cdot \log \delta^{-1} \cdot \log^{5} nd)$.
Adding the complexity for each part together, it is not hard to see the total space/time complexity for preprocessing
is  $O_{\epsilon,\gamma}(d n^{1+\rho} \cdot \log \delta^{-1} \cdot \text{Polylog}(nd))$ for some $0<\rho<1$ which depends on
$\epsilon, \gamma$.

\textbf{Query Time.} The query algorithm consists of constant number of nearest neighbor search and constrained AIFP queries.
Each nearest neighbor search takes $O_{c}(dn^{1/c} \cdot \text{Polylog}(n))$
where $c := 1 + \lambda$ and the constant hidden by
the $O_{c}(\cdot)$ notation depends only on $c$.
For constrained AIFP queries, 
note each bucket has no more than $n$ nodes,
from Lemma \ref{lm-aifp},
the query time is $O_{\epsilon,\gamma}(d n^{\rho} \cdot \log \delta^{-1} \cdot \log(nd))$
for some constant $0<\rho<1$ which depends on $\epsilon, \gamma$.
To summarize,
the query time is 
$O_{\epsilon,\gamma}(d n^{\rho} \cdot \log \delta^{-1} \cdot \text{Polylog}(nd))$
for some constant $0<\rho<1$ which depends on $\epsilon, \gamma$.

Combining all the above arguments, Theorem \ref{thm-main-aipf} is proved.

\section{MEB Range Aggregate Query}
\label{sec-ameb}


Given any query ball $B$,
we find the AMEB of $P \cap B$ using an iterative 
algorithm by Badoiu and Clarkson \cite{badoiu2003-mebcore}.
Their algorithm was originally designed for finding an approximate MEB for a fixed point set $P$. 
With careful analysis we show that their approach, after some modifications, can still be used to 
find AMEB in any given range $B$.
Briefly speaking, our idea is to construct a small-size coreset of $P \cap B$.
The MEB of the coreset is then a $(1+\epsilon)$-approximate MEB of 
$P \cap B$. 
The algorithm selects the coreset in an iterative fashion. It starts with an
arbitrary point $p$ from $P$. 
Each iteration, it performs the following operation to add a point to the coreset:
(1) Compute an (approximate) MEB of the current coreset; 
(2) Identify the IFP in $P$ to the center of the current MEB, and add it to the coreset.
We show that after $O_{\epsilon}(\log n)$ iterations,
the MEB of the coreset is then  a $(1+\epsilon)$-AMEB of $P\cap B$.

Let $0<\epsilon, \gamma, \delta <1$ be any small constants.
Below we discuss the details of how to use an AIFP data structure built with appropriate
parameters to answer $(\epsilon,\gamma)$-AMEB
with success probability at least $1 - \delta$.

Let $\epsilon_{A} := \gamma_{A} := \min(\epsilon/18,\gamma/18)$.
Define parameters $\epsilon_0 = \epsilon^{2}/1600$
and $\epsilon' = \min((1-\epsilon_0)^{-1} - 1, (1-\epsilon^2/100)^{-1/2} - 1, \epsilon/3)$.
Let $\mathcal{S}$ be a data structure that is
capable of answering $(\epsilon_A,\gamma_A)$-AIFP with success probability
at least $1 - \delta \epsilon^2_0/16$.
We show that with $\mathcal{S}$
and a $(1+\gamma_A)$-nearest neighbor data structure,
given any query ball $B$ (with query success probability $1 - \delta/2$),
the following Algorithm \ref{alg-queryAMEB}
output an AMEB with desired approximation quality and success probability.

The algorithm follows the main idea of the coreset algorithm in \cite{badoiu2003-mebcore}:
We first construct an initial coreset $P_{init}$.
This can be done by finding an arbitrary point in $p_1 \in B(1+\gamma)$,
then find another point $p_2$ which is an AIFP of $p_1$ in $P \cap B$,
then let $P_{init} := \{ p_1,p_2 \}$.
We compute the MEB of $P_{init}$; In each iteration, we find the approximate farthest point in $P \cap B$ to the center of the previous MEB, add it to the coreset, and then update the MEB of the coreset; We output the last found MEB if after one update, the size of MEB remains roughly the same. 

\begin{algorithm}[th]
\caption{AMEB-Query$(B)$}

\textbf{Input:} A query ball $B$ centered at $o_B$ with radius $r_B$. \\
\textbf{Output:} A ball $B_{MEB}$, which is a $(\epsilon, \gamma)$-AMEB of $P \cap B$.
\label{alg-queryAMEB}

\begin{algorithmic}[1]

\State{
Find a $(1+\gamma_A)$-nearest neighbor of $o_B$ in $P$ and let $p_a$ be the result.
If $p_a \not\in B(1+\gamma_A)$, \textbf{return} NULL.
}

\State{
Find an $(\epsilon_A,\gamma_A)$-AIFP of $q$ in $P \cap B$ and let $p_b$ be the result.
If $p_b$ is NULL, \textbf{return} NULL. Otherwise let $P_{init} := \{ p_a, p_b \}$
}

\State{ 
Initialize a set $P_{core}$ as $P_{init}$.
Compute a the MEB of $P_{core}$ (Note currently $P_{core}$ contains only two points so the MEB can be found easily).
Initialize $B_{0}$ to be the computed MEB.
Let $c_{0}$ be the center of $B_{0}$.
}

\State{ Let $w=\lceil 4 / \epsilon^2_0 \rceil$. For $i$ from $1$ to $w$, do the following: 
\begin{enumerate}
\item[4.1]
Try to find an $(\epsilon_{A},\gamma_{A})$-AIFP of $c_{i-1}$ in $B \cap P$
and denote the result as $p_{i}$.
Let $p'$ be the farthest point in $P_{core}$ to $c_{i-1}$.
If $p_{i}$ is NULL or $\norm{c_{i-1} - p_{i}} \leq \norm{c_{i-1} - p'}$,
\textbf{return} $B_{i-1}((1-\epsilon_A)^{-1})$ as the result.
\item[4.2] Add $p_i$ to $P_{core}$.
\item[4.3] Compute a $(1+\epsilon')$ approximate MEB of $P_{core}$ and denote the result as $B_{i}$.
(We use the approximate MEB algorithm in \cite{kumar22computing}.)
Let $c_{i}$ be the center of $B_{i}$, and $r_i$ be the radius of $B_{i}$.
\item[4.4] 
Let $r_i$ and $r_{i-1}$ be the radius of $B_i$ and $B_{i-1}$, respectively.
If $r_i \leq (1+\epsilon_0) r_{i-1}$,
\textbf{return} $B_i(1+\epsilon/3)$ as the result. 
\end{enumerate}
}

\State{ \textbf{return} $B_w(1+\epsilon/3)$ as the result. }

\end{algorithmic}
\end{algorithm}

We show that the query algorithm output an $(\epsilon,\gamma)$-AMEB of $B \cap P$
with probability at least $1-\delta$. This proves the correctness of the algorithm.

\begin{proof}
In the following we assume that the $w$ AIFP queries
and the nearest neighbor query 
are all successful. By simple calculation, we know that 
 this happens with probability at least $1-\delta$.

We first define some notations.
Let $P_{i}$ denote the set $P_{core}$ after the $i$-th iteration in Step 2 of Algorithm 
\ref{alg-queryAMEB} and $P_0 = P_{init}$.
Let $P'_i = P_i \cup (P \cap B)$.

Next, we show that the output ball is always an $(\epsilon, \gamma)$-AMEB of $P$ in $B$.
We consider 5 cases, depending on where the algorithm returns: (1) Step 1, (2) Step 2, (3) Step 4.1, (4) Step 4.4 and
(5) Step 5.

\noindent\textbf{Case (1):} If $p_a$ is a $(1+\gamma_A)$-nearest neighbor and $p_a \not\in B(1+\gamma_A)$,
this implies $P \cap B = \emptyset$. In this case, NULL is a valid answer for an AMEB of $P \cap B$.

\noindent\textbf{Case (2):} This case also only happens when $P \cap B = \emptyset$,
thus NULL is a valid answer.

\noindent\textbf{Case (3):} In this case, for every point $p \in B \cap P$, we have $p \in B_{i-1}((1-\epsilon_A)^{-1})$. Also,  it is clear that $P'_{i} \subset B(1+\gamma_A)$.
Thus, $B_{i-1}((1-\epsilon_A)^{-1})$ is 
a $((1+\epsilon')(1-\epsilon_A)^{-1})$-approximate MEB of $P'_{i}$, where $P \cap B \subseteq P'_{i} \subseteq P \cap B(1+\gamma_A)
\subseteq P \cap B(1+\gamma)$. 
By simple calculation, we know that 
 $B_{i-1}((1-\epsilon_A)^{-1})$ is a $(\epsilon,\gamma)$-AMEB of $P$ in $B$.

\noindent\textbf{Case (4):} To show this case, we need 
the following claim.
Let $i$ be any positive integer, and $q_i$ be the point on the ray $(c_{i-1},c_{i})$ whose distance to $c_{i-1}$ is $\norm{c_{i-1} - q_i} = \sqrt{1 - (1+\epsilon')^{-2}} \, r_{i-1}$.
Assume that $\norm{c_{i-1}-c_{i}} \geq \norm{c_{i-1} - q_i}$ (The case where $\norm{c_{i-1}-c_{i}} < \norm{c_{i-1} - q_i}$ 
will be discussed later in the proof).
Let $H_i$ be the hyperplane that passes through $c_{i}$
and is orthogonal to $(c_{i},c_{i-1})$. Let $\mathcal{H}_i$ be the closed half-space
bounded by $H_i$ and containing $c_{i-1}$.
Let $B^-(q_i,(1+\epsilon')^{-1} r_{i-1})$ be an \emph{open} ball
centered at $q_i$ and with radius $(1+\epsilon')^{-1} r_{i-1}$.

\begin{claim} \label{claim-11}
There exists a point $p'_{i-1} \in P_{i-1}$ such that
$p'_{i-1} \in \mathcal{H}_i \setminus B^-(q_i,(1+\epsilon')^{-1} r_i)$.
\end{claim}

\begin{figure}[htbp]
\centering
\includegraphics[width=0.7\textwidth]{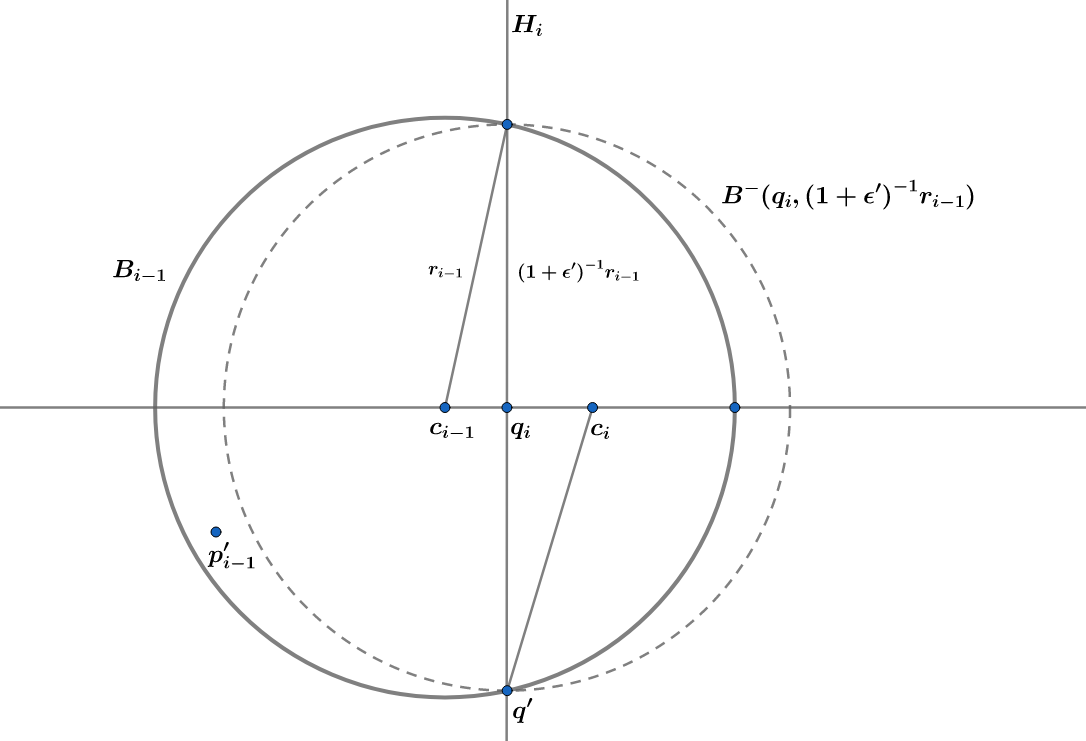}
\caption{An illustration for Claim \ref{claim-11} and its related arguments.}
\label{fig-circles}
\end{figure}

\emph{Proof of Claim \ref{claim-11}.} See Figure \ref{fig-circles} for a better understanding of the configuration.
Assume by contradiction that the claim does not hold. Then, it is not hard to
see that $B^-(q_i,(1+\epsilon')^{-1} r_i)$ contains every point in $P_{i-1}$.
Since $B^-(q_i,(1+\epsilon')^{-1} r_i)$ is an open ball, it is possible to slightly shrink the ball
to obtain a closed enclosing ball of $P_{i-1}$ whose radius is strictly smaller than $(1+\epsilon')^{-1} r_{i-1}$.
However, since $B_{i-1}$ is a $(1+\epsilon')$-approximate MEB of $P_{i-1}$ of radius $r_{i-1}$,
any enclosing ball of $P_{i-1}$ would have a radius at least $(1+\epsilon')^{-1} r_{i-1}$. This is a contradiction, and thus the claim is true. \qed

We first show that 
$\norm{c_{i} - c_{i-1}} \leq (\epsilon/10) r_{i-1}$.
To demonstrate this, we first consider the case where 
$\norm{c_{i} - c_{i-1}} \geq \sqrt{1-(1+\epsilon')^{-2}}  \, r_{i-1} = \norm{c_{i-1} - q_i}$ . 
Let $x_i = \norm{c_{i} - c_{i-1}} / r_{i-1}$, and  $p'_{i-1}$ be the point that satisfies the above claim.
Then, we have  
\begin{align}
\norm{p'_{i-1} - c_{i}}^2 
& \geq (1+\epsilon')^{-2} r^2_{i-1} + 
(\norm{c_{i} - c_{i-1}} - \sqrt{1-(1+\epsilon')^{-2}}  \, r_{i-1})^{2} \\
& = [(1+\epsilon')^{-2} + (x_i - \sqrt{1-(1+\epsilon')^{-2}})^2 ] r^2_{i-1}.
\end{align}

Since $B_{i}$ is an approximate MEB of $P_{i}$ and
$p'_{i-1} \in P_{i-1} \subset P_{i}$, we have
$\norm{p'_{i-1} - c_{i}} \leq r_{i} \leq (1+\epsilon_0) r_{i-1}$.
Combining this with the above inequality, we get
\[
(1+\epsilon_0)^2 \geq (1+\epsilon')^{-2} + (x_i - \sqrt{1-(1+\epsilon')^{-2}})^2. 
\]
From the definition of $\epsilon'$ and $\epsilon_0$
and some calculation,
we know that $x_i = \sqrt{(1+\epsilon_0)^2 - (1+\epsilon')^{-2}} + \sqrt{1-(1+\epsilon')^{-2}} 
\leq 2\sqrt{(1+\epsilon_0)^2 - (1+\epsilon')^{-2}}
\leq 2\sqrt{(1+\epsilon_0)^2 - (1-\epsilon_0)^{2}}
= 4\sqrt{\epsilon_0}
\leq \epsilon/10$.
This means that 
$\norm{c_{i} - c_{i-1}} \leq (\epsilon/10) r_{i-1})$
when $\norm{c_{i} - c_{i-1}} \geq \sqrt{1-(1+\epsilon')^{-2}}  \, r_{i-1}$.
For the case $\norm{c_{i} - c_{i-1}} < \sqrt{1-(1+\epsilon')^{-2}}  \, r_{i-1} = \norm{c_{i-1}-q}$, 
from the definition of $\epsilon'$, we have 
$\epsilon' \leq (1-\epsilon^2/100)^{-1/2} - 1$. 
Thus, $\sqrt{1-(1+\epsilon')^{-2}} \leq \epsilon/10$.
Combining the two cases, we conclude that $\norm{c_{i} - c_{i-1}} \leq (\epsilon/10) r_{i-1}$.

Let $p$ be an arbitrary point in $P \cap B$.
Then, we have $\norm{p_i-c_{i-1}} \geq (1-\epsilon_A) \norm{p - c_{i-1}}$.
Thus, $\norm{p - c_{i}} \leq 
\norm{p - c_{i-1}} + \norm{c_{i} - c_{i-1}}
\leq (1-\epsilon_A)^{-1} \norm{p_i - c_{i-1}} + (\epsilon/10) r_{i-1}
\leq (1-\epsilon_A)^{-1} \norm{p_i - c_{i}} + (1-\epsilon_A)^{-1} \norm{c_i - c_{i-1}} +
(\epsilon/10) r_{i-1}
\leq (1-\epsilon_A)^{-1} r_i + (1+(1-\epsilon_A)^{-1}) (\epsilon/10) r_{i-1}$.
Note that since $B_{i}$ is an enclosing ball of $P_i$,
it is also an enclosing ball of $P_{i-1}$.
From the fact that $B_{i-1}$ is a $(1+\epsilon')$-approximate MEB of $P_{i-1}$,
we know that $r_i \geq (1+\epsilon')^{-1} r_{i-1}$.
Plugging it into the above inequality,
we have 
$\norm{p - c_{i}} \leq 
[(1-\epsilon_A)^{-1} +  (1+(1-\epsilon_A)^{-1})(1+\epsilon') (\epsilon/10)] r_{i}$.
From the definition of $\epsilon_A,\epsilon'$,
we get that $\norm{p - c_{i}} \leq (1+\epsilon/3)r_i$.
This means that
for any $p \in B \cap P$,
$p \in B_i(1+\epsilon/3)$.
Since $B_i$ is a $(1+\epsilon/3)$-approximate MEB of $P_{i}$,
$B_i(1+\epsilon/3)$ is a 
$(1+\epsilon/3)^2$-approximate MEB of $P'_{i} = P_i \cup (B \cap P)$, and thus a $(1+\epsilon)$-approximate MEB.
Clearly, we have $B \cap P \subseteq P'_{i} \subseteq B(1+\gamma) \cap P$.
Therefore, $B_i(1+\epsilon/3)$ is a $(\epsilon,\gamma)$-AMEB of $P$ in $B$.

\noindent\textbf{Case (5):}
Note that for any $i = 0,1,\ldots,w-1$, we have 
$r_{i+1} \geq (1+\epsilon_0) r_i$.
Also note $w=\lceil 4 / \epsilon^2_0 \rceil$.
Thus we have $r_w \geq (1+\epsilon_0)^{w} r_0
\geq 4 r_0 / \epsilon_0$.
For any $p' \in B \cap P$,
we claim that
$p' \in B_w(1+\epsilon/3)$.
To see this, note clearly $r_0 = \norm{p_a - p_b}/2$,
and $(1-\epsilon_A)\norm{p' - p} \leq \norm{p_a - p_b}$.
Thus $\norm{p' - p} \leq 2 (1 - \epsilon_A)^{-1} r_0 \leq ((1 - \epsilon_A)^{-1}/2) \epsilon_0 r_w \leq \epsilon_0 r_w$.
Thus $p' \in B_w(1+\epsilon_0) \subset B_w(1+\epsilon/3)$.

Since $B_w$ is a $(1+\epsilon')$-approximate 
enclosing ball of a subset $P'$ of $P \cap B(1+\gamma_A) \subset P \cap B(1+\gamma)$, 
by simple calculation, we know that $B_w(1+\epsilon/3)$ is
an $(1+\epsilon)$-approximate MEB of $P'$.
Since $P \cap B \subset B_w(1+\epsilon/3)$,
we know $B_w(1+\epsilon/3)$ is
an $(1+\epsilon)$-approximate MEB of $P' \cup (P \cap B)$,
and clearly
$P \cap B \subset P' \cup (P \cap B) \subset P \cap B(1+\gamma)$.
Thus $B_w(1+\epsilon/3)$ is an $(\epsilon,\gamma)$-AMEB of $B \cap P$.

\end{proof}

\textbf{Complexity Analysis.} The nearest neighbor data structure takes $O_{\epsilon,\gamma}(dn^{1+\rho})$ time to build where $0<\rho<1$ depends on $\epsilon,\gamma$.
Algorithm \ref{alg-queryAMEB} makes one nearest neighbor search 
(which takes $O_{\epsilon,\gamma}(dn^{\rho})$ time),
and $O_{\epsilon,\gamma}(1)$ AIFP queries.
Combining the result of AIFP query (Theorem \ref{thm-main-aipf}),
Theorem \ref{thm-main} follows.

\FloatBarrier

\bibliographystyle{unsrt}
\bibliography{ref}

\newpage
\appendix

\section{Appendix}

\subsection{Proof of Lemma \ref{lm-aifp-correctness}}

\begin{proof}
In the following we assume that all the BD queries are successful. By simple calculation, we know that this happens with
probability at least $1-\delta$.

We first note that the return value would not be NULL if $(B,q)$ satisfies the constraint $(r_{B}, d_{min}, d_{max})$,
since this only happens when the BD query $(B,B_{out,0})$ returns a NULL,
which implies there is no point in $P$ that lies in $B \setminus B_{out,0}$.
This means there is no point in $P \cap B$ whose distance to $q$ is at least $r_0 = d_{min}$,
contradicting our assumption.

Next we consider the case where the algorithm returns at Step 3.
Clearly $p_{ans} \in B(1+\xi) \setminus B_{out,m}((1+\xi)^{-1})$.
Thus $\norm{p_{ans} - q} \geq d_{max}/(1+\xi)$.
Note for any $p \in P \cap B$, $\norm{q - p} \leq d_{max}$ from the constraint. 
The proof of this case is therefore complete.

Finally we consider the case that the algorithm returns at Step 2.
There exists $i$ such that $p_{ans} \in B(1+\xi) \setminus B_{out,i}((1+\xi)^{-1})$,
and the BD query $(B,B_{out,i+1})$ returns a NULL.
We know that there does not exist a point in $P \cap B \setminus B_{out,i+1}$.
Thus, let $p$ be any point in $P \cap B$,
we have $p \in B_{out,i+1}$, which means
$\norm{q - p} \leq r_{i+1} = (1+\xi)r_{i}$.
Since $p_{ans} \not\in B_{out,i}((1+\xi)^{-1})$,
we have $\norm{p_{ans} - q} \geq (1+\xi)^{-1} r_i$.
Thus $\norm{p_{ans} - q} \geq (1+\xi)^{-2} \norm{p - q} \geq (1-\epsilon) \norm{p - q}$.
This completes the proof.
\end{proof}

\subsection{Missing Proofs in Section \ref{sec-bd}}
\label{sec-bd-full}

The following is the proof of Lemma \ref{lm-bd-2}.

\begin{proof}
In order for $G_p$ to be examined, for every $j=1,2,\ldots,b$, the following two (independent) events have to happen:
(1) $\text{COM}(\textsc{lab}_{in}(p,j),\textsc{lab}_{in}(o_{in},j))$ $\geq t_1 \wedge \text{COM}(\textsc{lab}_{out}(p,j)$, $\textsc{lab}_{out}(o_{out},j)) \geq t_2$,
 which happens with probability at least $4/9$,
 and (2) $\textsc{lab}_{in}(p,j) = \textsc{lab}'_{in}(j)$ and $\textsc{lab}_{out}(p,j) = \textsc{lab}'_{out}(j)$, which happens with probability $2^{-2a}$.
Thus, the probability of the above two events happen is $[(4/9)(2^{-2a})]^b$. This proves the lemma for $G_p$.
The case for $G_q$ can also be proved in a similar manner.
\end{proof}

With the above two lemmas, we can obtain the following lemma. 
The argument has some similarity with some proofs in \cite{har2012approximate} for near neighbor search with LSH.

The following is the proof of Lemma \ref{lm-bd-3}.

\begin{proof}
For ease of analysis, we consider a modified version of Algorithm \ref{alg-queryBD}:
we assume that the algorithm does not terminate after $3c$ points are examined,
and it does not return even after a point in $B_{in}(1+\xi) \setminus B_{out}((1+\xi)^{-1})$
is found. Consider the following event:
\begin{itemize}
    \item $E$: Among the first $3c$ examined points, there exists one point that is in $P \cap (B_{in}(1+\xi) \setminus B_{out}((1+\xi)^{-1}))$.
\end{itemize}

We prove that event $E$ happens with probability at least $1/4$. The lemma will then follows.

Let $p \in P$ be a point that lies in $B_{in} \setminus B_{out}$.
Consider the following 2 events:
\begin{itemize}
    \item $E_1$: $p$ is examined.
    \item $E_2$: The number of points $q \not\in (B_{in}(1+\xi) \setminus B_{out}((1+\xi)^{-1}))$
    such that $q$ is examined is strictly less than $3c$.
\end{itemize}
If both of the events occur, then clearly $E$ happens. Thus, we only need to prove that 
with probability $1/4$ both $E_1$ and $E_2$ happen.

We first consider $E_1$. Note that for every $k=1,2,\ldots,c$,
the probability that $p$ is examined in the for loop of Step 1 of Algorithm \ref{alg-queryBD}
is at least $(P''_1)^b$. Note that
\[
(P''_1)^b \geq (P''_1)^{1 + \log_{1/P''_2} n} = P''_1 n^{-\rho}.
\]
The probability that $E_1$ does not happen is at most $(1 - (P''_1)^b)^c \leq (1 - P''_1 n^{-\rho})^c \leq
(1 - P''_1 n^{-\rho})^{n^{\rho}/P''_1} \leq 1/e$.
Thus, we have $\text{PR}[E_1] \geq 1 - 1/e$.

Next, we consider event $E_2$.
Let $q$ be any point that lies outside of $B_{in}(1+\xi) \setminus B_{out}((1+\xi)^{-1})$.
We call such $q$ a ``bad'' point.
For any bad point $q$,
the probability that $q$ is examined in the for loop of Step 1 of Algorithm \ref{alg-queryBD}
is at most $(P''_2)^b \leq 1/n$.
Thus, the expected number of examined bad point in each iteration of the loop is at most 1.
Let $X$ denote the number of examined bad points
after all the $c$ iterations.
Then, we have $E[X] \leq c$.
By Markov's inequality,
we get $\text{Pr}[X \geq 3c] \leq 1/3$.
Hence, $\text{Pr}[E_2] = \text{Pr}[X < 3c] \geq 2/3$.

The probability that both $E_1$ and $E_2$ happen is thus no less than $(1 - 1/e) + 2/3 - 1 \geq 1/4$.
This completes the proof.
\end{proof}

In the following we analyze the space/time complexity of the BD query scheme
and prove Lemma \ref{lm-bd}.
We first consider the process of creating the BD data structure in Algorithm \ref{alg-buildBD}.
The algorithm creates $c$ groups of buckets. For each group,
every point in $P$ is hashed into one  bucket.
In order to create these buckets, we sample a total of $2abc$ hash functions from LSH families $\mathcal{H}_{in}$ and $\mathcal{H}_{out}$, and
apply them to every point in $P$ to determine the buckets it belongs to.
The total running time for this is $O(dabcn)$. Since $a = O_{\xi}(1), b = O_{\xi}(\log n)$ and $c = O_{\xi}(n^\rho)$, where $\rho$ and the constant hidden in the big-O notation depends only on $\xi$.
Thus, the preprocessing time and the space complexity of the $\xi$-error BD query data structure is $O_{\xi}(dn^{1+\rho}\log n)$.
 
Next, we analyze the query process, {\em i.e.,} Algorithm \ref{alg-queryBD}.
The main loop has at most $c$ iterations. In each iteration,
we first determine the bucket to examine, which takes $O(dab)$ time to compute the labels.
Also, note that the process will terminate after $3c$ points are examined.
Thus, the total running time for the query algorithm is $O(abcd)$, which is $O_{\xi}(dn^{\rho}\log n)$.

Note that it is easy to see that a BD query data structure with query success rate $1-\delta$
can be obtained by combining $O(\log \delta^{-1})$ BD structures with query success rate $1/4$.
Lemma \ref{lm-bd} then follows, which concludes this subsection.

\subsection{Missing Proofs in Section \ref{sec-aifp}}

\noindent\textbf{Proof of Claim \ref{cl-far}.}
\begin{proof}
Indeed, we will prove a slightly stronger version of Claim \ref{cl-far}:
If $\norm{q - o_B} \geq (3+\gamma)\epsilon^{-1}r_B$,
then any point $p$ in $B(1+\gamma/2)$ is an $(\epsilon, \gamma/2)$-AIFP of $q$ in $P \cap B$.

Let $p'$ be arbitrary point in $P \cap B$,
then $\norm{p' - p} \leq (2+\gamma)r_B$.
Note $\norm{p' - q} \geq \norm{q - o_B} - \norm{p' - o_B}
\geq (3+\gamma)\epsilon^{-1}r_B - r_B \geq (2+\gamma) \epsilon^{-1} r_B$.
Therefore $\norm{p' - p} \leq \epsilon \norm{p' - q}$.
Thus $\norm{p - q} \geq \norm{p' - q} - \norm{p' - p} \geq (1-\epsilon) \norm{p' - q}$.
Since $p$ is in $B(1+\gamma/2)$,
we conclude that $p$ is an $(\epsilon, \gamma/2)$-AIFP of $q$ in $P \cap B$.

\end{proof}

\quad

\noindent\textbf{Proof of Claim \ref{cl-out}.}

\begin{proof}
If $v$ does not exists, then $r_N > \gamma r_B / 64$.
Since $p_N$ is a 2-nearest neighbor,
for any $p' \in P$,
we have $\norm{p' - q} \geq r_N/2$.
Thus $\norm{p' - q} > \gamma r_B / 128$,
and the claim follows easily in this case.

In the following, let $v_p$ denote the parent node of $v$ (if $v_p$ does not exists, \emph{i.e.} $v$ is the root
of $T_P$, then the claim is trivial).
Then from the definition of $v$,
we know that $s(v_p) + r_N \geq \gamma r_B/64$.
From the property of $T_P$,
for any $p_{out} \in P \setminus P(v)$,
$s(v_p) \leq  \norm{p_{out} - q} nd$.
Thus $\norm{p_{out} - q} nd + r_N \geq \gamma r_B/64$.
Therefore, either $r_N \geq \gamma r_B/128$ or
$\norm{p_{out} - q} nd  \geq \gamma r_B/128$.
Note that if $r_N \geq \gamma r_B/128$,
then the claim follows easily,
since $p_N$ is a 2-nearest neighbor of $q$ in $P$.
So we assume $\norm{p_{out} - q} nd  \geq \gamma r_B/128$,
which gives us $\norm{p_{out} - q} \geq \gamma r_B/128nd$.
The proof is complete.

\end{proof}

\quad

\noindent\textbf{Proof of Claim \ref{cl-case1}}

\begin{proof}
Assume $\norm{o_B - q} > (1+\gamma/64)r_B = r_B + \gamma r_B/64$.
Note $s(v) + r_N \leq \gamma r_B/64$,
and $s(v)$ is an upper bound of the diameter of $P(v)$,
then for any $p \in P(v)$, we have $\norm{p - q} \leq \gamma r_B/64$.
Therefore $p \not\in B$.
\end{proof}

\quad

\noindent\textbf{Proof of Claim \ref{cl-case2}}

\begin{proof}
We first consider case $r^{-} = r_N + s(v) > \gamma r_B /512 n^2 d^2$.
Then either (a) $r_N > \gamma r_B /1024 n^2 d^2$, or (b) $s(v) > \gamma r_B /1024 n^2 d^2$.
For case (a), let $p_1$ be arbitrary point in $P \cap B \subset P \cap \mathcal{B}(o_B, r^{(2)}_B)$,
since $p_N$ is a 2-nearest neighbor,
we have $\norm{p_1 - q} \geq r_N/2 \geq \gamma r_B /2048 n^2 d^2$.
For case (b),
note $s(v)$ is not larger than $nd$ times the diameter of $P(v)$,
thus there exists $p_1, p_2 \in P(v)$,
such that $s(v) \leq \norm{p_1 - p_2} nd$.
Thus $\norm{p_1 - p_2} \geq \gamma r_B /1024 n^3 d^3$.
Thus at least one of $p_1$ or $p_2$, just say $p_1$,
would satisfy $\norm{p_1 - q} \geq \gamma r_B /2048 n^3 d^3$.
Also $p_1 \in \mathcal{B}(o_B, r^{(2)}_B)$. In fact, 
for any point $p' \in P(v)$,
$\norm{o_B - p'} \leq \norm{o_B - q} + \norm{q - p'} \leq (1+\gamma/64)r_B + (r_N + s(v)) \leq (1+\gamma/64)r_B + \gamma r_B/64
= (1+\gamma/32)r_B$,
which implies $p' \in \mathcal{B}(o_B, r^{(2)}_B)$.
We have proved item 1 of the claim.

Now assume $r^{-} = r_N + s(v) \leq \gamma r_B /512 n^2 d^2$.
For every $p \in P(v)$,
we have $\norm{p - q} \leq r_N + s(v) = r^{-} \leq [r^{-}]_{\lambda}$,
thus $p \in \mathcal{B}(o_B, r^{(3)}_B)$.
Also from Claim \ref{cl-out},
for any $p_{out} \in P \setminus P(v)$,
we have $\norm{q-p_{out}} \geq \gamma r_B /256nd$.
From this, it is not hard to see by calculation that
$p_{out}$ lies outside of $p \in \mathcal{B}(o_B, (1+\gamma) r^{(3)}_B)$.

Note clearly $s(v) + r_N = r^{-} \geq (1+\lambda)^{-1} r^{(3)}_B \geq r^{(3)}_B/2$.
If $r_N \geq r^{(3)}_B/4$,
then we have $\norm{q - p_N} = r_N \leq r^{(3)}_B/4$ and $p_N \in P(v)$,
and the claim is proved in this case.
Assume $r_N < r^{(3)}_B/4$,
then $s(v) \geq r^{(3)}_B/4$.
Now let $p_1, p_2 \in P(v)$ such that $\norm{p_1 - p_2}$ is the diameter of $P(v)$.
Thus $r^{(3)}_B/4 \leq s(v) \leq nd \norm{p_1 - p_2}$.
Thus $\norm{p_1 - p_2} \geq r^{(3)}_B/4nd$.
Then for at least one of $p_1,p_2$,
its distance to $q$ is at least $r^{(3)}_B/8nd$.
The proof is complete.

For any $p' \in \mathcal{B}(o_B, r^{(3)}_B)$,
$\norm{o_B - p'} \leq \norm{o_B - q} + \norm{q - p'} \leq (1+\gamma/64)r_B + \gamma r_B /256 n^2 d^2
\leq (1+\gamma/32)r_B$. Thus $\mathcal{B}(o_B, r^{(3)}_B) \subset B(1+\gamma/32)$.

\end{proof}

\end{document}